%% file: main.tex
\pdfoutput=1
\documentclass{article}

\usepackage[final,main]{neurips_2026}

\usepackage[utf8]{inputenc}
\usepackage[T1]{fontenc}
\usepackage[hypertexnames=false]{hyperref}
\usepackage{url}
\usepackage{booktabs}
\usepackage{amsfonts}
\usepackage{amsmath}
\usepackage{amssymb}
\usepackage{amsthm}
\usepackage{mathtools}
\usepackage{nicefrac}
\usepackage{microtype}
\usepackage{xcolor}
\usepackage{graphicx}
\usepackage{algorithm}
\usepackage{algpseudocode}
\usepackage{bm}

\hypersetup{pdftitle={Counterfactual Rollout Replay: Forkable Environments as Free Process Rewards for Software Engineering Agents},pdfauthor={Yuanhao Li, Hongbo Wang, Xuhong Chen, Yiming Cao, Xunzhu Tang},colorlinks=true,linkcolor=blue,citecolor=blue,urlcolor=blue}

\newtheorem{proposition}{Proposition}

\title{Counterfactual Rollout Replay:\\Forkable Environments as Free Process Rewards\\for Software Engineering Agents}

\author{%
Yuanhao Li$^{1}$ \quad Hongbo Wang$^{1}$\thanks{Corresponding author: \texttt{hbwang@bupt.edu.cn}.} \quad Xuhong Chen$^{1}$\\
\textbf{Yiming Cao$^{1}$ \quad Xunzhu Tang$^{2}$}\\[3pt]
$^{1}$State Key Laboratory of Networking and Switching Technology\\
Beijing University of Posts and Telecommunications, Beijing 100876, China\\
$^{2}$University of Luxembourg, Luxembourg
}

\begin{document}

\maketitle

\input{sections/abstract}
\input{sections/introduction}
\input{sections/related-work}
\input{sections/related-work-additions}
\input{sections/preliminaries}
\input{sections/method}
\input{sections/experiments}
\input{sections/discussion}
\input{sections/conclusion}

\begin{ack}
This work received no third-party funding or support. The authors declare no competing interests.
\end{ack}

\label{page:references}
\bibliographystyle{plainnat}
\bibliography{refs}

\clearpage
\appendix
\input{sections/appendix-proofs}
\input{sections/appendix-algorithm}
\input{sections/appendix-snr}
\input{sections/appendix-hyperparams}
\input{sections/appendix-prompts}
\input{sections/appendix-determinism}
\input{sections/appendix-extra-ablations}
\input{sections/appendix-perrepo}
\input{sections/appendix-error-taxonomy}
\input{sections/appendix-cases}
\input{sections/appendix-cost-breakdown}
\input{sections/appendix-compute}

\input{sections/appendix-additional-discussion}
\input{sections/appendix-dataset}
\input{sections/appendix-reproducibility}
\input{sections/appendix-extended}

\newpage
\input{sections/checklist}

\end{document}

%% file: sections/abstract.tex
\begin{abstract}
Outcome-only reinforcement learning gives software engineering (SWE) agents a terminal success signal but little direct guidance about intermediate decisions. We introduce Counterfactual Rollout Replay (CRR), a training-time procedure that uses forkable executable environments to obtain step-level return contrasts. CRR selects a small set of decision points, restores each state, samples an alternative action, and rolls the branch forward under the policy. It retains the realized training trajectory and replaces the advantage at selected steps with the difference between its terminal return and the sampled counterfactual return. The method needs no human process labels or learned process reward model; \emph{free} refers to those supervision costs, not replay compute. With a 14B policy, CRR improves pass@1 on \textsc{SWE-bench Verified}, \textsc{SWE-bench Live}, and \textsc{SWE-rebench}, and combines with process-reward and trajectory-search methods. On \textsc{SWE-bench Verified}, an equal-wall-clock comparison on the same hardware yields 41.7\% versus 36.7\% for extended vanilla (outcome-only) GRPO, a 5.0-point gain with fork overhead included. These results apply to environments with affordable, reliable state restoration; stochastic continuations and expensive or imperfect replay remain limitations.
\end{abstract}

%% file: sections/introduction.tex
\section{Introduction}
\label{sec:intro}

Software engineering agents solve repository-level issues through sequences of file inspection, code edits, shell commands, and tests. Benchmarks such as \textsc{SWE-bench}~\citep{jimenez2023swebench} and executable training environments such as \textsc{SWE-Gym}~\citep{pan2024swegym} and \textsc{SWE-rebench}~\citep{badertdinov2025swerebench,badertdinov2026swerebench2} make terminal-verifier reinforcement learning (RL) practical. Proximal policy optimization (PPO) and group relative policy optimization (GRPO) updates~\citep{schulman2017ppo,shao2024deepseekmath} can improve task success, but a trajectory-level advantage does not distinguish a useful localization decision from a later edit that breaks the patch. Learned process rewards and search provide alternative sources of guidance~\citep{cui2025prime,han2026swetrace,peng2026hiper,djuhera2026tsr}. We investigate whether executable alternatives can instead provide a training-time credit signal.

In SWE, forkability provides access to a restorable repository and executable environment with auditable provenance. With sandbox infrastructure~\citep{pan2024swegym,wang2024openhands,yuan2026minisandbox}, an alternative action can be followed through the real task environment and evaluated by its verifier, without learning a transition model. This opportunity depends on reliable restoration and an affordable replay cost.

\emph{Counterfactual Rollout Replay} (CRR) selects a small number of decision points from an on-policy trajectory, forks their states, and samples alternative actions and continuations. It retains the realized trajectory for training and uses the difference between realized and counterfactual terminal returns as the advantage at selected steps. Unlike selecting a verifier-best branch for training, the fork changes how the realized action is scored. \emph{Free} means that CRR needs no human process labels, learned process reward model, or oracle hindsight annotation.

Our contributions are:
\begin{itemize}
\item \textbf{Executable contrast estimation.} CRR supplies an environment-grounded selected-step signal and integrates it into the PPO/GRPO update while retaining the realized trajectory.
\item \textbf{Analysis.} The estimator is unbiased for a proposal-dependent, action-conditional contrast at a fixed index. The implemented trajectory-level selector induces a selector-weighted objective, and the variance decomposition provides a diagnostic condition rather than a universal guarantee.
\item \textbf{Controlled evaluation.} With a shared 14B backbone, CRR improves pass@1 across \textsc{SWE-bench Verified}, \textsc{SWE-bench Live}, and \textsc{SWE-rebench}, and combines with process-reward models (PRMs) and trajectory-search training. A separate equal-wall-clock comparison on Verified retains a 5.0-point gain over extended vanilla GRPO with replay overhead charged. A matched Search-Select control tests the two uses of fork data under a reduced budget.
\end{itemize}

%% file: sections/related-work.tex
\section{Related Work}
\label{sec:related}

\paragraph{Repository-level SWE agents and benchmarks.} Tool-using agents such as ReAct established the reasoning-and-acting loop for language models~\citep{yao2023react}, and \textsc{SWE-bench} reframed this loop as repository-level GitHub issue resolution~\citep{jimenez2023swebench}. Subsequent systems improved the agent-computer interface, search pipeline, and multi-agent scaffold: \textsc{SWE-agent} tailored file editing and test execution to autonomous agents~\citep{yang2024sweagent}; \textsc{AutoCodeRover}, \textsc{MAGIS}, \textsc{LingmaAgent}, \textsc{HyperAgent}, and \textsc{OpenHands} developed increasingly capable repository-level scaffolds and platforms~\citep{zhang2024autocoderover,tao2024magis,ma2024lingma,phan2024hyperagent,wang2024openhands}. Benchmarking has also moved toward cleaner and fresher evaluation through \textsc{SWE-bench Verified}~\citep{openai2024swebenchverified}, \textsc{SWE-bench Live}~\citep{zhang2025swebenchlive}, \textsc{SWE-bench+}~\citep{aleithan2024swebenchplus}, multilingual or re-collected SWE tasks~\citep{badertdinov2025swerebench,badertdinov2026swerebench2,zan2024swebenchjava}, and adversarial or audit-style reassessments~\citep{wang2025solvedissues,yu2026sweabs}. Complementary microbenchmarks probe dependency-version constraint semantics with machine-checkable answers~\citep{chen2026semverbenchbenchmarkingllmcomprehension}. CRR is complementary to this line: it changes the training-time advantage estimator rather than the scaffold, benchmark, or verifier.

\paragraph{RL with verifiable rewards (RLVR) for software engineering agents.} Recent SWE-agent training systems move beyond prompting by optimizing code policies against executable feedback. \textsc{RLEF} grounds code policies in execution feedback~\citep{gehring2024rlef}, \textsc{SWE-Gym} provides thousands of executable Python tasks for training agents and verifiers~\citep{pan2024swegym}, and recent pipelines such as \textsc{SoRFT}, \textsc{SWE-Fixer}, \textsc{SWE-Lego}, \textsc{Murphy}, \textsc{SkyRL-Agent}, \textsc{SWE-Dev}, and \textsc{SWE-MiniSandbox} refine supervised fine-tuning (SFT), RL objectives, and infrastructure efficiency~\citep{ma2025sorft,xie2025swefixer,tao2026swelego,ekbote2025murphy,cao2025skyrl,du2025swedev,yuan2026minisandbox}. Fine-grained studies of supervised fine-tuning also identify cases of incomplete learning~\citep{xue2026supervised}. The underlying optimization machinery follows PPO-style clipped policy gradients~\citep{schulman2017ppo} and GRPO-style group-relative advantages introduced for mathematical reasoning and scaled in RLVR systems~\citep{shao2024deepseekmath,guo2025deepseekr1}. These pipelines make training possible, but their reward signal remains sparse and trajectory-level; CRR targets this credit-assignment bottleneck.

\paragraph{Process rewards and multi-turn credit assignment.} One response to sparse terminal rewards is to train or infer dense process feedback. Process-supervised reward models~\citep{lightman2024verify}, implicit rewards such as PRIME~\citep{cui2025prime}, rubric-based SWE PRMs such as \textsc{SWE-Trace}~\citep{han2026swetrace}, execution-free feedback models~\citep{shum2025swerm}, and self-evolved reasoning rewards~\citep{guan2025rstarmath} all provide denser supervision than a terminal pass bit. Structured auxiliary tasks provide another form of process supervision for multimodal decisions~\citep{liu-etal-2026-ips}, while generative reward modeling can allocate reasoning according to internal uncertainty~\citep{xue2026reason}. For evaluation with large language models (LLMs), \citet{feng2026secondorderresponselawsllm} distinguish within-prompt sampling noise from between-prompt variation in LLM judges and derive debiased estimators of quadratic prompt instability. A parallel line imposes structure on long-horizon trajectories through hierarchical agent RL~\citep{peng2026hiper,zhou2024archer}, hindsight or belief-update credit assignment~\citep{tan2026hindsight,auzina2026belief}, and trajectory-search rollouts~\citep{djuhera2026tsr}. Scoring-based selection of multilingual reasoning paths has also been used to construct fine-tuning data~\citep{weihua2026adamcot}. CRR differs mechanistically: it does not learn a process reward model, infer credit from text, or select the best branch. It executes a counterfactual branch in the real SWE environment and uses the observed return differential as a training-time advantage signal.

\paragraph{Counterfactual and causal credit assignment.} Counterfactual credit assignment has a pre-LLM history in RL. Future-conditional baselines separate skill from luck by learning hindsight-conditioned value functions~\citep{mesnard2021counterfactual}, while counterfactual data augmentation and causal approaches manufacture alternative transitions under structural assumptions~\citep{pitis2020coda,lu2020counterfactual}. Counterfactual reasoning has also been studied in multi-agent communication and information-theoretic credit assignment~\citep{vanneste2020counterfactual,arumugam2021infotheoretic}, and recent LLM-agent work adapts contextual counterfactual credit to multi-agent collaboration~\citep{chen2026contextual}. These works generally learn, model, or infer counterfactual outcomes; CRR exploits the forkability of sandboxed SWE tasks to execute selected counterfactuals directly.

\paragraph{Forking, sandboxing, and inference-time search.} Execution feedback also supports LLM-driven optimization of programmatic policies and other artifacts~\citep{kuang2025learning,nie2026understanding}. The environmental property we use is not forkability in the abstract: games and simulators may also restore state, but SWE repositories can be restored from a commit and executable environment with auditable provenance. Containerized SWE infrastructures such as \textsc{SWE-Gym} and \textsc{OpenHands}, and the container-free \textsc{SWE-MiniSandbox}, make this practical during training~\citep{pan2024swegym,wang2024openhands,yuan2026minisandbox}, and CRR can be viewed as a Dyna-style use of the real environment as the model rather than a learned simulator~\citep{sutton1991dyna}. Inference-time tree search and test-time scaling for coding agents also branch environments, but they use branches to choose better completions or rollout continuations~\citep{koh2024treesearch,kim2026agenticcoding,zhu2025agenticscaling,barnes2026surprisal,djuhera2026tsr}. The interplay between search algorithms and reward design has also been surveyed~\citep{wei2025unifying}. CRR uses forks differently: a branch is replayed to estimate what the realized action contributed relative to a plausible alternative, not to replace the behavior policy's chosen action.

%% file: sections/related-work-additions.tex
\paragraph{Tree-based training and uncertainty-guided branching.}
Branching also supports training-time credit assignment: Tree-GRPO shares rollout prefixes and combines intra-tree and inter-tree group-relative advantages derived from outcome rewards~\citep{ji2025treegrpo}.
This is distinct from inference-time best-first search for selecting action trajectories~\citep{koh2024treesearch} and from Entropy-Guided Branching, which replays alternative actions at high-entropy steps to repair failed tool-use trajectories~\citep{wei2026egb}.
CRR instead retains the realized on-policy trajectory and uses environment-executed alternative-action contrasts to score selected decisions during training.
Its contribution is therefore this credit-assignment construction in forkable SWE environments, rather than branching or entropy-based selection alone.

%% file: sections/preliminaries.tex
\section{Preliminaries}
\label{sec:prelim}

We model issue resolution as a finite-horizon Markov decision process $\mathcal{M}=(\mathcal{S},\mathcal{A},P,R,H)$. State $s_t$ includes the repository working tree, issue, conversation history, and tool observations. Actions are structured tool calls for file inspection, editing, shell execution, and testing~\citep{yang2024sweagent}; executing a call in the sandbox induces $P(s_{t+1}\mid s_t,a_t)$~\citep{pan2024swegym,yuan2026minisandbox}. We assume deterministic transitions for the admitted training tasks and describe the consistency filter and residual nondeterminism in Section~\ref{sec:appendix-determinism}. The policy $\pi_\theta$ samples a trajectory $\tau=(s_0,a_0,\ldots,s_T,a_T)$ with $T<H$, i.e., at most $H$ actions (default $H=50$). Its terminal reward $R(\tau)\in\{0,1\}$ records whether the final patch passes the held-out tests. We write $Q^{\pi}(s,a)$ for the expected terminal return after taking action $a$ in state $s$ and then following $\pi$, and $V^{\pi}(s)=\mathbb{E}_{a\sim\pi(\cdot\mid s)}[Q^{\pi}(s,a)]$.

For expected return $J(\theta)=\mathbb{E}_{\tau\sim\pi_\theta}[R(\tau)]$ and an action-independent baseline $b$, the policy-gradient identity is
\begin{equation}
\nabla_\theta J(\theta)=\mathbb{E}_{\tau\sim\pi_\theta}\left[\sum_{t=0}^{T}(R(\tau)-b)\nabla_\theta\log\pi_\theta(a_t\mid s_t)\right].
\label{eq:pg}
\end{equation}
For the GRPO comparison~\citep{shao2024deepseekmath,guo2025deepseekr1}, we write the centered outcome advantage as $\hat{A}^{\mathrm{GRPO}}_t=R(\tau)-\bar{R}^{(G)}$, where $\bar{R}^{(G)}$ is the mean return of $G$ sibling trajectories from the same initial state. This scalar is shared across the realized trajectory.

A sandbox is \emph{forkable at horizon $h$} if it supports a snapshot $\sigma:\mathcal{S}\to\Sigma$ and restore $\rho:\Sigma\to\mathcal{S}$ such that executing any fixed sequence of $h$ actions from $\rho(\sigma(s_t))$ has the same trajectory distribution as executing it from $s_t$. In the deterministic case, the resulting states are identical. The SWE setting provides a restorable executable repository state; whether snapshotting and replay are cheap enough to help training is an empirical resource question (Section~\ref{sec:exp-compute}).

A counterfactual trajectory $\tilde{\tau}_t$ shares $\tau$'s prefix through $s_t$, samples $\tilde{a}_t\sim q_{\neg a_t}(\cdot\mid s_t)$ with $\tilde{a}_t\ne a_t$, and then continues with independently sampled actions from $\pi_\theta$. When the policy assigns near-unit probability to one action, we skip that index rather than force an out-of-support alternative. Operationally, the selector's entropy gate $\mathcal{H}_t>\eta$ with $\eta=0.4$ (Table~\ref{tab:hparams}) already excludes such near-deterministic steps, since a near-unit action probability drives the masked-token entropy of the emitted call toward zero.

%% file: sections/method.tex
\section{Counterfactual Rollout Replay}
\label{sec:method}

\subsection{The CRR Advantage Estimator}
\label{sec:method-estimator}

Given an on-policy trajectory $\tau$, the selector chooses decision indices $\mathcal{I}(\tau)\subseteq\{0,\ldots,T\}$. For each $t\in\mathcal{I}(\tau)$, CRR restores $s_t$, samples $\tilde{a}_t\sim q_{\neg a_t}(\cdot\mid s_t)$, and completes an independently sampled on-policy continuation. The selected-step advantage is
\begin{equation}
\hat{A}^{\text{CRR}}_t=R(\tau)-R(\tilde{\tau}_t).
\label{eq:crr}
\end{equation}
CRR retains the realized trajectory rather than replacing it with a better branch (Figure~\ref{fig:overview}). At unselected steps it uses the outcome-only advantage, giving
\begin{equation}
\hat{A}^{\star}_t=
\begin{cases}
R(\tau)-R(\tilde{\tau}_t) & t\in\mathcal{I}(\tau),\\
R(\tau)-\bar{R}^{(G)} & \text{otherwise}.
\end{cases}
\label{eq:hybrid}
\end{equation}
We insert this hybrid advantage into the clipped surrogate:
\begin{equation}
\mathcal{L}_{\text{CRR}}(\theta)=
\mathbb{E}_{\tau\sim\pi_{\theta_{\text{old}}}}\!\left[
\sum_{t=0}^{T}\min\!\left(
r_t(\theta)\hat{A}^{\star}_t,\,
\mathrm{clip}(r_t(\theta),1-\epsilon,1+\epsilon)\hat{A}^{\star}_t
\right)\right]
-\beta D_{\mathrm{KL}}(\pi_\theta\|\pi_{\text{ref}}),
\label{eq:loss}
\end{equation}
where $r_t(\theta)=\pi_\theta(a_t\mid s_t)/\pi_{\theta_{\text{old}}}(a_t\mid s_t)$ and $\beta$ weights the Kullback--Leibler (KL) penalty toward a reference policy $\pi_{\text{ref}}$. The advantages $\hat{A}^{\star}_t$ are treated as constants in Eq.~\eqref{eq:loss}; no gradient is propagated through the counterfactual return differential. The ratio is formed per decision point: $\pi_\theta(a_t\mid s_t)$ is the product of the token-level probabilities of the emitted tool call, and clipping is applied to this per-action ratio. This retains the PPO/GRPO update structure~\citep{schulman2017ppo,shao2024deepseekmath}, without implying a new stability or convergence guarantee. The fixed-index contrast below is unbiased, but selecting indices from an entire realized trajectory induces data-dependent weighting. The implemented objective is selector-weighted, not an unbiased estimator of the unweighted per-step GRPO objective.

\begin{figure}[t]
\centering
\includegraphics[width=0.78\linewidth]{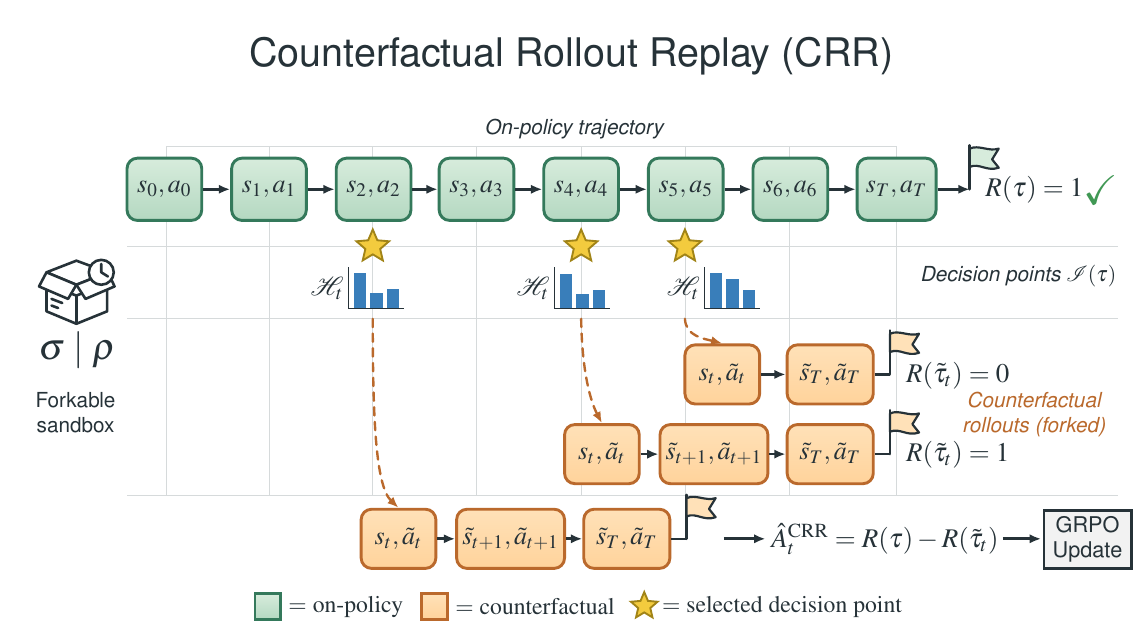}
\caption{CRR retains the realized trajectory (green), forks selected decision points, and samples alternative continuations (orange). Their terminal-return differences replace outcome-only advantages at selected indices.}
\label{fig:overview}
\end{figure}

\subsection{Properties}
\label{sec:method-properties}

Proofs appear in Section~\ref{sec:appendix-proofs}.
\begin{proposition}[Fixed-index action-conditional contrast]
\label{prop:unbiased}
Assume a deterministic sandbox forkable at horizon $H-t$ and a proposal $q_{\neg a_t}(\cdot\mid s_t)$ supported on alternative actions $\mathcal{A}\setminus\{a_t\}$. For a fixed index $t$ and action-augmented history $x_t=(s_0,a_0,\ldots,s_t,a_t)$,
\[
\mathbb{E}\!\left[\hat{A}^{\text{CRR}}_t\mid x_t\right]
=Q^{\pi_\theta}(s_t,a_t)
-\mathbb{E}_{\tilde{a}_t\sim q_{\neg a_t}}\!\left[Q^{\pi_\theta}(s_t,\tilde{a}_t)\right].
\]
The expectation includes the counterfactual action and future policy-sampling randomness in both continuations. This is an action-conditional contrast for the specified proposal. Only the idealized variant allowing $\tilde{a}_t=a_t$ with $q=\pi_\theta$ recovers the classical advantage by reducing the second term to $V^{\pi_\theta}(s_t)$.
\end{proposition}

\begin{proposition}[Variance decomposition against an outcome baseline]
\label{prop:variance}
Write $X=R(\tau)$, $Y=R(\tilde{\tau}_t)$, and let $\bar{Y}=\bar{R}^{(G)}$ be the sibling-group baseline. For any conditioning variable or bucket $B$ with finite second moments,
\begin{align*}
&\mathrm{Var}[X-Y\mid B]-\mathrm{Var}[X-\bar{Y}\mid B]\\
&\quad=\mathrm{Var}[Y\mid B]-\mathrm{Var}[\bar{Y}\mid B]
-2\mathrm{Cov}(X,Y\mid B)+2\mathrm{Cov}(X,\bar{Y}\mid B).
\end{align*}
CRR therefore has lower conditional variance than the outcome baseline in buckets satisfying
\[
\mathrm{Cov}(X,Y\mid B)>
\frac12\bigl(\mathrm{Var}[Y\mid B]-\mathrm{Var}[\bar{Y}\mid B]\bigr)
+\mathrm{Cov}(X,\bar{Y}\mid B).
\]
\end{proposition}

The decomposition gives a diagnostic condition, not universal variance reduction.

\subsection{Decision-Point Selector}
\label{sec:method-selector}

We select at most $k$ non-terminal decisions per trajectory (default $k=4$). Let $\mathcal{H}_t$ denote the policy's uncertainty over the action at step $t$; we approximate it by the mean grammar-masked next-token entropy over semantic tokens of the emitted tool call. The tool-type score is 0 for passive reads, 0.25 for focused tests, 0.75 for shell commands (including search), and 1.0 for file edits; terminal submit actions are excluded. Define
\begin{equation}
\mathrm{score}(t)=\mathcal{H}_t\bigl(1+\lambda_{\text{type}}\mathrm{type}(a_t)\bigr).
\label{eq:selector}
\end{equation}
The selector takes the top-$k$ scores among steps with $\mathcal{H}_t>\eta$, using fewer forks if fewer steps qualify. It combines uncertainty with a preference for state-changing actions. For ablations, Random samples up to $k$ indices uniformly without replacement from the same eligible set; Stride selects approximately evenly spaced time-ordered quantiles of that set. All use the same proposal and fork cap.

\subsection{Counterfactual Proposal}
\label{sec:method-proposal}

We use a tempered policy proposal excluding the realized action:
\[
q_{\neg a_t}(a\mid s_t)\propto
\pi_{\theta,T_{\text{cf}}}(a\mid s_t)\,\mathbf{1}[a\ne a_t],
\qquad T_{\text{cf}}>0,
\]
where $\pi_{\theta,T}$ is $\pi_\theta$ decoded token by token at temperature $T$ under the same JSON-grammar constraint; sampling from $\pi_{\theta,T_{\text{cf}}}$ and rejecting $a=a_t$ draws from $q_{\neg a_t}$ exactly.
The default $T_{\text{cf}}=0.7$ balances alternative-action diversity against plausible continuations. The reported sweep is broadly insensitive over $[0.5,1.0]$. The zero-contrast control displayed at $T_{\text{cf}}=0$ is a distinct control outside this positive-temperature proposal family; it is not its zero-temperature limit.

\subsection{Implementation}
\label{sec:method-implementation}

The implementation uses \textsc{SkyRL-Agent}'s asynchronous rollout framework~\citep{cao2025skyrl} and \textsc{SWE-MiniSandbox}~\citep{yuan2026minisandbox}. Algorithm~\ref{alg:crr} gives the training loop. On-policy generation and counterfactual replay run in logical worker pools on the same node, synchronizing before the update. The default uses $G=8$ sibling trajectories and up to $k=4$ selected steps per trajectory, with one counterfactual continuation per selected step. Section~\ref{sec:exp-compute} separates trajectory counts from total device-time and reports the equal-wall-clock comparison with all fork overhead charged.

%% file: sections/experiments.tex
\section{Experiments}
\label{sec:experiments}

\subsection{Experimental Setup}
\label{sec:exp-setup}

\paragraph{Model and infrastructure.}
We train the open-weight \textsc{Qwen3-14B} backbone~\citep{yang2025qwen3} with rank-128 low-rank adaptation (LoRA) adapters, using \textsc{SkyRL-Agent}~\citep{cao2025skyrl} for asynchronous dispatch and \textsc{SWE-MiniSandbox}~\citep{yuan2026minisandbox} for execution. Training uses one node with eight H100 80GB GPUs; fork workers are logical roles whose device-time is charged to this node. The default setting uses $k=4$ selected steps, one counterfactual per step, $T_{\mathrm{cf}}=0.7$, and learning rate $5\times10^{-6}$. Full-budget runs collect 24,000 on-policy trajectories and perform 4,000 optimizer updates. Sections~\ref{sec:appendix-hyperparams} and~\ref{sec:appendix-prompts} provide hyperparameters, prompts, and tools.

\paragraph{Data and evaluation.}
The post-filter training corpus contains 3,472 deterministic instances: 2,272 admitted \textsc{SWE-Gym} tasks~\citep{pan2024swegym} and 1,200 admitted \textsc{SWE-rebench} tasks~\citep{badertdinov2025swerebench}. These counts follow the three-rerun filter (Section~\ref{sec:appendix-determinism}). The \textsc{SWE-rebench} training/evaluation split precedes filtering and excludes overlap by task identity, issue/pull-request URL, collection month, repository commit window, exact or near-duplicate patch hunks, and near-duplicate issue text. We evaluate greedy pass@1 with a 50-turn cap on \textsc{SWE-bench Verified}~\citep{openai2024swebenchverified}, \textsc{SWE-bench Live}~\citep{zhang2025swebenchlive}, and held-out \textsc{SWE-rebench} Python tasks. \textsc{SWE-bench Live} uses the frozen June-2025 snapshot (1{,}565 instances across 164 repositories, issues created after January 2024), and the held-out \textsc{SWE-rebench} evaluation set is the 860-task Python leaderboard split (instances collected after April 2025), disjoint from the training split by the rules above. Main results are means and standard deviations across three random seeds (three repeated evaluations for the untrained zero-shot agents).

\paragraph{Baselines and resource control.}
Our same-backbone reproductions cover zero-shot agents~\citep{yang2024sweagent,wang2024openhands}, SFT~\citep{pan2024swegym,xie2025swefixer}, outcome-only and multi-turn RL~\citep{gehring2024rlef,ekbote2025murphy,cao2025skyrl}, process supervision (\textsc{SWE-Trace})~\citep{han2026swetrace}, and trajectory search (TSR)~\citep{djuhera2026tsr}. Matching on-policy trajectories does not match total compute. For the central 24,000-trajectory comparison, GRPO, \textsc{SWE-Trace}, and CRR consume 188, 233, and 321 H100-hours, respectively (Table~\ref{tab:main-cost}). Table~\ref{tab:main} therefore reports endpoint performance, while the separate equal-wall-clock experiment measures resource efficiency.

\begin{table}[t]
\centering\small
\caption{Full-budget resource totals for the three central training comparisons. Device-time includes on-policy rollout, counterfactual replay or \textsc{SWE-Trace} scoring, and learner updates. These settings do not have equal total compute. Wall-clock values are rounded to one decimal place.}
\label{tab:main-cost}
\begin{tabular}{lrrr}
\toprule & GRPO & \textsc{SWE-Trace} & CRR\\
\midrule Device-time (H100-hours) & 188 & 233 & 321\\
Node wall-clock (hours) & 23.5 & 29.1 & 40.1\\
Verified pass@1 (\%) & $36.4\pm0.8$ & $39.1\pm0.7$ & $41.7\pm0.6$\\
\bottomrule
\end{tabular}
\end{table}

\subsection{Main Results}
\label{sec:exp-main}

CRR reaches 41.7\%, 35.9\%, and 32.6\% on Verified, Live, and \textsc{SWE-rebench} (Table~\ref{tab:main}), exceeding vanilla GRPO by 5.3, 4.8, and 4.7 percentage points. The corresponding differences from \textsc{SWE-Trace} are 2.6, 2.2, and 2.4 points. CRR also improves over TSR+GRPO, while combining CRR with TSR or \textsc{SWE-Trace} raises performance further. The 44.9\% Verified result uses CRR plus \textsc{SWE-Trace}; the full CRR+TSR+\textsc{SWE-Trace} combination reaches 45.6\% (Table~\ref{tab:composition}).

\begin{table}[t]
\centering
\caption{pass@1 (\%) with the same 14B backbone. Values are means $\pm$ seed standard deviations over three seeds; for the two zero-shot rows, which involve no training, they are means $\pm$ standard deviations over three repeated greedy evaluations and reflect inference and sandbox nondeterminism. Bold denotes the best non-stacked result; underlining denotes the best result within this table. The final column is the Verified difference from vanilla GRPO.}
\label{tab:main}
\small
\setlength{\tabcolsep}{4pt}
\begin{tabular}{lcccc}
\toprule
Method & \shortstack{\textsc{SWE-bench}\\\textsc{Verified}} & \shortstack{\textsc{SWE-bench}\\\textsc{Live}} & \textsc{SWE-rebench} & $\Delta$ vs.\ GRPO\\
\midrule
\textsc{SWE-agent} (zero-shot)        & $22.6\pm 1.4$ & $18.9\pm 1.5$ & $16.4\pm 1.6$ & $-13.8$\\
\textsc{OpenHands} (zero-shot)        & $23.4\pm 1.4$ & $19.7\pm 1.5$ & $17.0\pm 1.6$ & $-13.0$\\
\textsc{SWE-Gym} SFT                  & $28.1\pm 1.2$ & $24.0\pm 1.3$ & $21.5\pm 1.4$ & $-8.3$\\
\textsc{SWE-Fixer}                    & $30.7\pm 1.1$ & $25.8\pm 1.3$ & $22.8\pm 1.4$ & $-5.7$\\
\textsc{RLEF} (outcome-only PPO)      & $33.5\pm 1.0$ & $28.6\pm 1.2$ & $25.4\pm 1.3$ & $-2.9$\\
\textsc{Murphy}                       & $35.2\pm 0.9$ & $30.0\pm 1.1$ & $26.7\pm 1.2$ & $-1.2$\\
\textsc{SkyRL-Agent} (vanilla GRPO)   & $36.4\pm 0.8$ & $31.1\pm 1.1$ & $27.9\pm 1.2$ & $\phantom{+}0.0$\\
\textsc{SWE-Trace} (rubric PRM)       & $39.1\pm 0.7$ & $33.7\pm 1.0$ & $30.2\pm 1.1$ & $+2.7$\\
TSR + GRPO                            & $38.0\pm 0.8$ & $32.5\pm 1.0$ & $29.1\pm 1.1$ & $+1.6$\\
\midrule
CRR (ours)                            & $\bm{41.7\pm 0.6}$ & $\bm{35.9\pm 0.9}$ & $\bm{32.6\pm 1.0}$ & $+5.3$\\
CRR + TSR                             & $43.6\pm 0.5$ & $37.6\pm 0.8$ & $34.1\pm 0.9$ & $+7.2$\\
CRR + \textsc{SWE-Trace}              & $\underline{44.9\pm 0.4}$ & $\underline{38.8\pm 0.8}$ & $\underline{35.4\pm 0.9}$ & $+8.5$\\
\bottomrule
\end{tabular}
\end{table}

\subsection{Sample Efficiency and Fully Charged Training Time}
\label{sec:exp-sample}
\label{sec:exp-compute}

Figure~\ref{fig:sample-efficiency} separates two resource axes. On the left, CRR already exceeds the full-budget vanilla-GRPO pass-rate at its first observed equal-wall-clock checkpoint: 37.2\% after 8,704 on-policy trajectories, versus 36.4\% at vanilla GRPO's 24,000-trajectory endpoint (Table~\ref{tab:main}) and 36.7\% for the extended GRPO run at 37,824. This is a $2.8\times$--$4.3\times$ reduction in the counted on-policy trajectories; counterfactual suffixes are excluded, so this is not a compute saving.

The full accounting charges 164 H100-hours for on-policy rollout, 133 for the counterfactual pool, and 24 for learner updates: 321 H100-hours, or 40.125 hours on eight H100 GPUs (Table~\ref{supp-accounting}). Overlap between fork workers and generation is not deducted from device-time. The right panel compares continuing runs on identical hardware at equal elapsed training time. At 40.125 hours, CRR reaches 41.7\% versus extended GRPO's 36.7\%, a 5.0-point advantage. GRPO has then completed 37,824 trajectories and 6,304 updates, compared with CRR's 24,000 and 4,000. Timing includes forks, tests, waits, retries, optimization, and checkpoints; the complete ledger and timing exclusions appear in Table~\ref{supp-wallclock} and Section~\ref{supp-equal-time}.

\begin{figure}[t]
\centering
\includegraphics[width=0.49\linewidth]{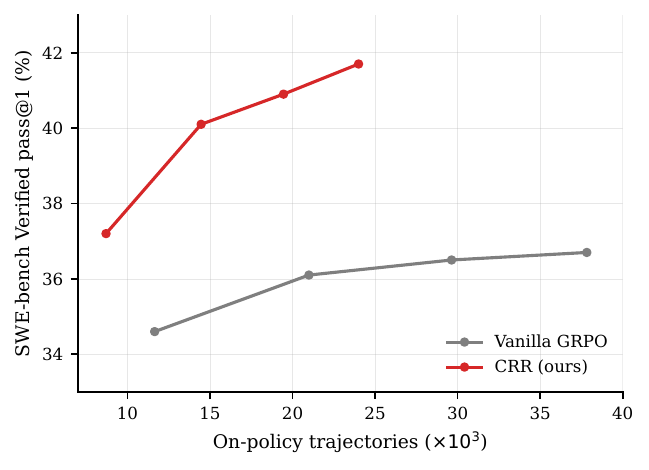}\hfill
\includegraphics[width=0.49\linewidth]{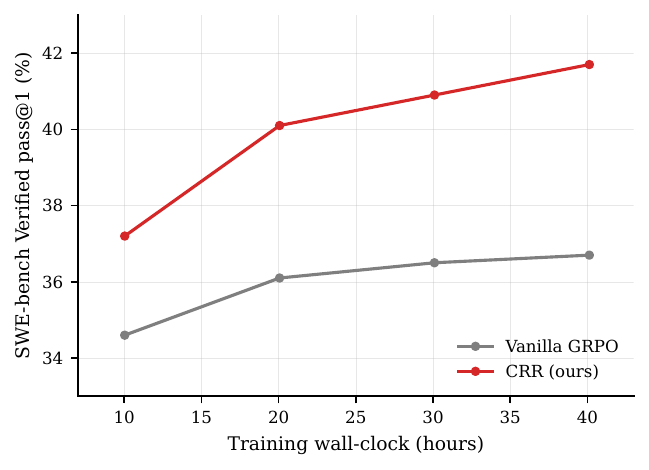}
\caption{Verified performance against two distinct budgets. Left: newly collected on-policy trajectories at the observed equal-wall-clock checkpoints of Table~\ref{supp-wallclock} (three-seed means, no interpolation), excluding counterfactual suffixes. Right: elapsed training time on the same $8\times$H100 node with replay overhead fully charged. Per-checkpoint uncertainty bands are not available.}
\label{fig:sample-efficiency}
\end{figure}

\paragraph{Operating range and signal diagnostics.}
A reduced-budget sweep allocates 48 H100-hours per arm and evaluates all 500 Verified instances with one seed. Multiplying snapshot/restore/replay cost by $m=1,2,4$ gives CRR 36.2, 34.8, and 33.2\%, versus GRPO's 33.4\%. The observed advantage disappears between $m=2$ and $m=4$. Advantage signal-to-noise ratio (SNR) and task/bucket covariance diagnostics are reported in Section~\ref{sec:appendix-snr} and Figure~\ref{fig:snr}.

\subsection{Credit Assignment versus Search-Select}
\label{sec:exp-search-select}

To distinguish credit assignment from selecting better branches, we compare CRR with a Search-Select-GRPO control. Both use the same selector, $k=4$ forks per realized trajectory, verifier, horizon, group size, and clipped learner. Search-Select replaces a realized training sequence only when a fork has strictly higher verifier return; ties retain the realized sequence. It then applies ordinary GRPO, discarding unused branches. CRR retains the realized sequence and uses the return differential for credit assignment. Under equal device-time, CRR is higher than Search-Select-GRPO (Table~\ref{tab:search-select}), supporting this mechanism distinction within the tested protocol. Section~\ref{sec:supplementary-evaluation} gives the branch-selection and log-probability accounting.

\begin{table}[t]
\centering
\small
\caption{Two uses of matched fork executions. Every arm receives 48 H100-hours per seed; greedy evaluation covers all 500 Verified instances. Values are means $\pm$ seed standard deviations over three seeds.}
\label{tab:search-select}
\begin{tabular}{lcr}
\toprule
Method & Fork use & pass@1 (\%)\\
\midrule
Vanilla GRPO & None & $33.3\pm0.7$\\
Search-Select-GRPO & Sequence selection & $34.4\pm0.8$\\
CRR & Credit assignment & $36.2\pm0.6$\\
\bottomrule
\end{tabular}
\end{table}

\subsection{Ablations and Composition}
\label{sec:exp-ablations}

Figure~\ref{fig:ablations} examines four design choices. Increasing $k$ from zero to one gains 2.2 points; $k=4$ recovers 5.3 of the 5.6-point gain at $k=16$, with additional wall-clock cost. The entropy-plus-tool-type selector improves over either component alone. For truncated horizons $h\in\{2,4,8\}$, the verifier evaluates the fork's current working tree; full-horizon rollouts continue to submission or the 50-turn cap. Performance is relatively stable for $T_{\mathrm{cf}}\in[0.5,1.0]$ and degrades at higher temperatures. Section~\ref{sec:appendix-extra-ablations} covers concurrency and the PPO outer loop.

\begin{figure}[t]
\centering
\includegraphics[width=0.95\linewidth]{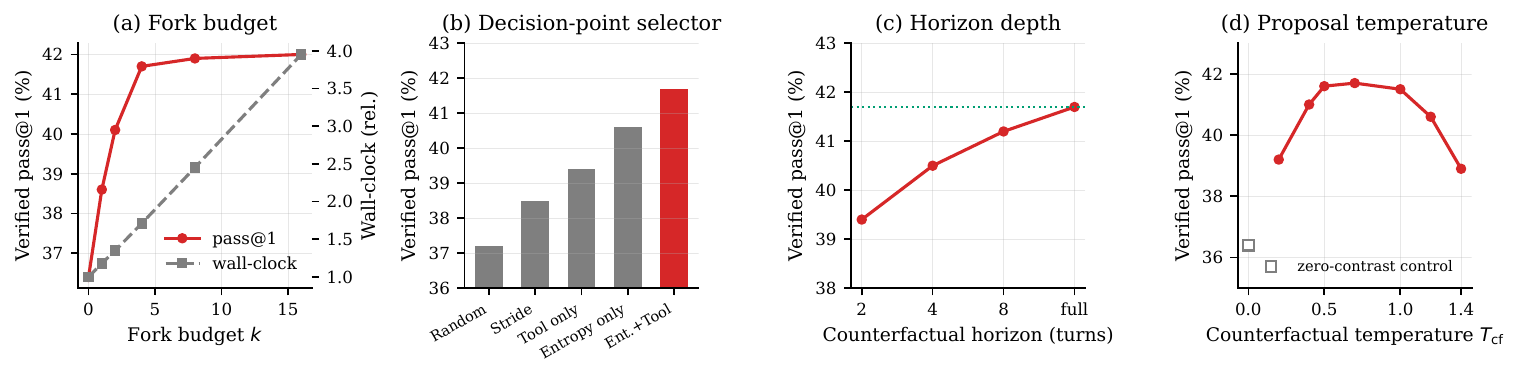}
\caption{Verified ablations: (a) fork budget and its relative wall-clock cost, not an equal-wall-clock comparison; (b) selector; (c) counterfactual horizon; (d) proposal temperature, with a separate zero-contrast control plotted at zero.}
\label{fig:ablations}
\end{figure}

Table~\ref{tab:composition} shows gains when CRR is combined with TSR and \textsc{SWE-Trace} supervision. Adding CRR raises the four configurations by 5.3, 5.6, 5.8, and 5.6 points, respectively. The full stack reaches 45.6\%. These endpoint comparisons demonstrate compatible mechanisms.

\begin{table}[t]
\centering
\caption{Composition on Verified, pass@1 (\%). TSR and \textsc{SWE-Trace} (a rubric PRM) are added to base GRPO; the 44.9\% cell is CRR+\textsc{SWE-Trace}, and 45.6\% adds TSR. Bold marks the CRR row; the three-component stack is underlined instead.}
\label{tab:composition}
\small
\begin{tabular}{lcccc}
\toprule
Configuration & none & TSR & \textsc{SWE-Trace} & TSR+\textsc{SWE-Trace}\\
\midrule
Base GRPO & $36.4$ & $38.0$ & $39.1$ & $40.0$\\
Base GRPO + CRR & $\bm{41.7}$ & $\bm{43.6}$ & $\bm{44.9}$ & $\underline{45.6}$\\
\bottomrule
\end{tabular}
\end{table}

\subsection{Sampling Consistency and Counterfactual Quality}
\label{sec:exp-robustness}
The gain persists beyond greedy evaluation. Table~\ref{tab:sampling-main} evaluates eight sampled completions per instance at temperature 0.6 using the three full-budget checkpoints. CRR improves pass@8 by 3.9 points over GRPO and all-8 success by 5.2 points. all-8 is the fraction of instances solved by every sampled completion: it measures consistency across samples. The full pass@1/2/4/8 results are in Section~\ref{sec:supplementary-evaluation}.

\begin{table}[!t]
\centering\small
\caption{Sampling evaluation on all 500 Verified instances using eight completions per instance (temperature $0.6$). Reported percentages are averaged over three full-budget training seeds. Greedy pass@1 is included for context.}
\label{tab:sampling-main}
\begin{tabular}{lrrrr}
\toprule Method & Greedy pass@1 & Sampled pass@1 & pass@8 & all-8\\
\midrule Vanilla GRPO & 36.4 & 35.9 & 53.8 & 18.7\\
\textsc{SWE-Trace} & 39.1 & 38.6 & 55.6 & 21.2\\
CRR & 41.7 & 41.2 & 57.7 & 23.9\\
\bottomrule
\end{tabular}
\end{table}

The appendix further tests how the fork budget is allocated and whether counterfactual quality explains the gain. Among 20,000 logged forks, 12\% use an immediately invalid alternative; this class contributes 8\% of a class-mean plug-in contrast diagnostic. A validity-filtered proposal changes pass@1 from 36.2\% to 36.4\% in the reduced-budget, single-seed comparison. Selector-alignment probes find enrichment for patch-relevant edits and large sampled contrasts. Controlled replay-noise experiments show gradual degradation in the tested setting. Section~\ref{sec:supplementary-evaluation} gives the full protocols and diagnostics.

%% file: sections/discussion.tex
\section{Discussion}
\label{sec:discussion}

\paragraph{Limitations.} CRR requires affordable, reliable forks and a trustworthy terminal verifier. Expensive restoration can erase its resource-efficiency advantage, and the reported operating points do not isolate severe trajectory-length skew. The reduced-budget noise experiments test a controlled injection protocol; averaging multiple continuations does not remove systematic replay bias. Deterministic transitions also leave future policy-sampling variance intact. Our unbiasedness result concerns a fixed-index contrast, whereas the implemented selector induces a weighted training objective without an established convergence guarantee. The selector-alignment diagnostics provide decision-relevance proxies, not proof of per-step causality. Additional boundary cases and cost diagnostics appear in Section~\ref{sec:appendix-additional-discussion}.

\paragraph{Broader impact.} Better SWE agents can support software maintenance, while executable counterfactual traces make their training signals more inspectable. The same capabilities can accelerate malicious software development, and verifier-coupled training can exploit incomplete tests or harnesses. CRR does not resolve these risks. We recommend adversarial test-suite hardening~\citep{yu2026sweabs}, execution-free corroboration such as \textsc{SWE-RM}~\citep{shum2025swerm}, and documentation of replay traces and verifier limitations when releasing training artifacts.

%% file: sections/conclusion.tex
\paragraph{Conclusion.} CRR uses executable counterfactual continuations to score selected actions while retaining the realized training trajectory. It requires neither process labels nor a learned reward model, but charges the additional replay compute. The reported experiments show improved pass@1 across three SWE benchmarks and an equal-wall-clock gain over vanilla GRPO on Verified. Its practical value depends on affordable, reliable forkability; understanding replay noise, selector weighting, and heterogeneous trajectory costs remains an important direction.

%% file: sections/appendix-proofs.tex
\section{Proofs}
\label{sec:appendix-proofs}

\begin{proof}[Proof of Proposition~\ref{prop:unbiased}]
Fix an action-augmented history $x_t = (s_0,a_0,\dots,s_t,a_t)$. Determinism of the sandbox fixes the environment transitions for any realized future action sequence, but the policy continuation after $a_t$ remains stochastic. Let
\[
Q^{\pi_\theta}(s,a)=\mathbb{E}_{a_{t+1:\,T}\sim\pi_\theta}\!\left[R(\tau)\mid s_t=s, a_t=a\right]
\]
denote the expected terminal return after taking action $a$ at state $s$ and then continuing with $\pi_\theta$. Taking expectations over both the counterfactual action $\tilde{a}_t\sim q_{\neg a_t}$ and the future policy continuation randomness gives
\[
\mathbb{E}\!\left[\hat{A}^{\text{CRR}}_t\mid x_t\right] = Q^{\pi_\theta}(s_t,a_t) - \mathbb{E}_{\tilde{a}_t\sim q_{\neg a_t}}\!\left[Q^{\pi_\theta}(s_t,\tilde{a}_t)\right].
\]
If one uses the idealized proposal $q=\pi_\theta$ without excluding the realized action, the second term is $V^{\pi_\theta}(s_t)$ and the contrast recovers the classical advantage. The implemented proposal conditions on $\tilde{a}_t\ne a_t$, so the exact estimator targets the alternative-action contrast stated in Proposition~\ref{prop:unbiased}.
\end{proof}

\begin{proof}[Proof of Proposition~\ref{prop:variance}]
Let $X = R(\tau)$, $Y = R(\tilde{\tau}_t)$, and $\bar{Y}$ a sibling-group baseline drawn from the same group. For any conditioning variable or bucket $B$,
\begin{align*}
\mathrm{Var}[X-Y\mid B] &= \mathrm{Var}[X\mid B] + \mathrm{Var}[Y\mid B] - 2\,\mathrm{Cov}(X,Y\mid B),\\
\mathrm{Var}[X-\bar{Y}\mid B] &= \mathrm{Var}[X\mid B] + \mathrm{Var}[\bar{Y}\mid B] - 2\,\mathrm{Cov}(X,\bar{Y}\mid B).
\end{align*}
Subtracting the two identities gives the decomposition in Proposition~\ref{prop:variance}; rearranging terms gives the stated sufficient condition. Section~\ref{sec:appendix-snr} reports prefix-bucket covariance and SNR diagnostics supporting the empirical regime in which this condition is plausible, but the decomposition itself does not assert universal variance reduction.
\end{proof}

%% file: sections/appendix-algorithm.tex
\section{Algorithm Pseudo-Code}
\label{sec:appendix-algorithm}

Algorithm~\ref{alg:crr} reproduces the per-iteration training loop of CRR in pseudo-code. It is implemented atop the asynchronous multi-turn rollout dispatcher of \textsc{SkyRL-Agent}~\citep{cao2025skyrl} with the container-free sandboxing of \textsc{SWE-MiniSandbox}~\citep{yuan2026minisandbox}. The fork pool runs as logical worker roles on the same node, in parallel with the main rollout pool; in our profile, fork/materialization and counterfactual rollout are the two largest overhead components (Section~\ref{sec:appendix-cost-breakdown}).

\begin{algorithm}[h]
\caption{Counterfactual Rollout Replay (per training iteration)}
\label{alg:crr}
\begin{algorithmic}[1]
\Require Policy $\pi_\theta$, task batch $\{i_1,\dots,i_N\}$, budget $k$, entropy threshold $\eta$, proposal $q_{\neg a}$
\Ensure Updated policy parameters $\theta'$
\For{each task $i_b$ in parallel}
  \State Sample $G$ on-policy trajectories $\{\tau^{(g)}\}_{g=1}^G$ from $\pi_\theta$ in the sandbox
  \For{$g=1,\dots,G$}
    \State Compute step scores via Eq.~\eqref{eq:selector}; let $\mathcal{I}_g$ be up to $k$ non-terminal indices with $\mathcal{H}_t>\eta$ (this gate also excludes near-deterministic steps)
    \For{$t\in\mathcal{I}_g$}
      \State Snapshot sandbox state $\sigma(s_t)$
      \State Sample $\tilde{a}_t\sim q_{\neg a_t}(\cdot\mid s_t)$ with $\tilde{a}_t\ne a_t$
      \State Restore $\rho(\sigma(s_t))$ and roll out $\tilde{\tau}_t$ on-policy from $\tilde{a}_t$
      \State Record $R(\tilde{\tau}_t)$
    \EndFor
  \EndFor
  \State Form hybrid advantages $\hat{A}^{\star}_t$ via Eq.~\eqref{eq:hybrid}
\EndFor
\State Update $\theta$ by minibatch gradient steps on the collected data (about 43 optimizer updates per iteration; Table~\ref{tab:hparams})
\end{algorithmic}
\end{algorithm}

%% file: sections/appendix-snr.tex
\section{Empirical Correlation Diagnostics}
\label{sec:appendix-snr}

Proposition~\ref{prop:variance} gives a variance decomposition and a sufficient covariance-and-baseline condition for CRR to reduce conditional variance relative to an outcome baseline. As a diagnostic for this regime, we measure whether on-policy and counterfactual returns covary positively within task/prefix buckets. For every training step at which counterfactual rollouts are collected, we bin paired returns by task id, selected-step decile, and a prefix signature derived from the first ten semantic tokens of the realized action history, then compute the within-bucket covariance and average across buckets weighted by bucket size (Table~\ref{tab:cov}).

\begin{table}[htbp]
\centering
\caption{Empirical task/prefix-bucket covariance $\hat{\mathrm{Cov}}(R(\tau), R(\tilde{\tau}_t)\mid B)$ at five checkpoints during training. The covariance remains strictly positive throughout, with the average value increasing as the policy becomes more confident in its prefix decisions.}
\label{tab:cov}
\small
\begin{tabular}{lccccc}
\toprule
Step                              & 200    & 1{,}000 & 2{,}000 & 3{,}000 & 4{,}000\\
\midrule
$\hat{\mathrm{Cov}}$ (mean)        & $0.044$ & $0.071$ & $0.103$ & $0.119$ & $0.128$\\
$\hat{\mathrm{Cov}}$ (5th pct.)    & $0.018$ & $0.029$ & $0.046$ & $0.059$ & $0.068$\\
$\hat{\mathrm{Cov}}$ (95th pct.)   & $0.094$ & $0.131$ & $0.171$ & $0.183$ & $0.198$\\
\bottomrule
\end{tabular}
\end{table}

The 5th-percentile covariance is strictly positive at every checkpoint. This does not by itself prove the sufficient condition in Proposition~\ref{prop:variance}, because the condition also depends on counterfactual variance and sibling-baseline covariance. It supports the intended regime diagnostically: as the policy improves, task/prefix buckets become more predictive of eventual outcome and counterfactual continuations from nearby prefixes end up in increasingly correlated terminal states.

The plotted per-step advantage-SNR diagnostic is computed within the same task/prefix buckets: for each bucket $b$ containing $n_b$ of the $N$ per-step advantage estimates of a method ($\hat A_t^\star$ for CRR), it divides the absolute mean $|\mu_b|$ by the standard deviation $\sigma_b$ and averages over buckets, $\mathrm{SNR}=\sum_b (n_b/N)\,|\mu_b|/\sigma_b$. It is recorded every 200 optimizer updates. It is an advantage-signal diagnostic; it does not directly measure the variance of the complete policy-gradient update.

\begin{figure}[htbp]
\centering
\includegraphics[width=0.75\linewidth]{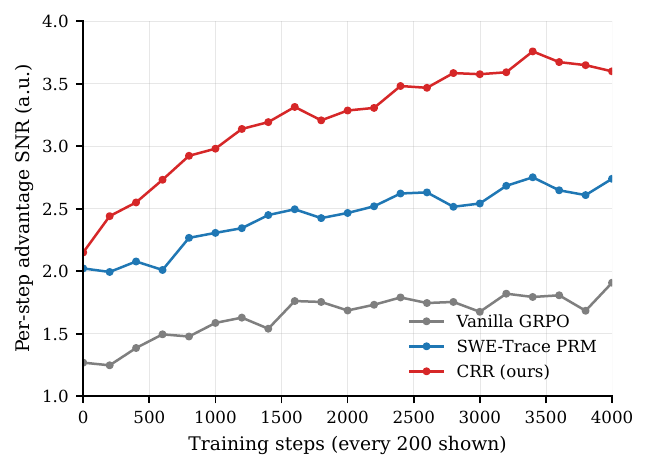}
\caption{Per-step advantage SNR across training; training steps are optimizer updates. CRR achieves approximately twice the advantage SNR of vanilla GRPO and about $1.4$ times that of the \textsc{SWE-Trace} PRM in the reported diagnostic. This advantage-SNR comparison does not establish the sufficient variance condition in Proposition~\ref{prop:variance}.}
\label{fig:snr}
\end{figure}

%% file: sections/appendix-hyperparams.tex
\section{Hyperparameters}
\label{sec:appendix-hyperparams}

Table~\ref{tab:hparams} lists the full-budget CRR configuration. We distinguish newly collected on-policy environment trajectories from optimizer minibatch updates over the collected rollout buffer. The reduced-budget controls use a fixed 48 H100-hour budget and may complete different numbers of trajectories and updates. The shared GRPO hyperparameters follow the \textsc{SkyRL-Agent} default configuration. The CRR-specific parameters ($k$, $K$, $T_{\text{cf}}$, $\eta$, $\lambda_{\text{type}}$) were set once by a single sweep on a 200-instance development subset drawn from the admitted \textsc{SWE-Gym} training pool and then frozen for all reported runs; the subset remains part of the 2{,}272 training tasks, and no evaluation task informed any hyperparameter choice. The \textsc{Verified} ablations are post-hoc diagnostics of these frozen choices. The parameters introduced by CRR are the fork budget $k$, the number of counterfactual samples per fork $K$, the counterfactual temperature $T_{\text{cf}}$, the entropy threshold $\eta$, and the type weight $\lambda_{\text{type}}$.

\begin{table}[htbp]
\centering
\caption{Default hyperparameters for the full-budget CRR runs. Ablations vary their stated factor, and equal-device-time controls may complete fewer trajectories or updates.}
\label{tab:hparams}
\small
\begin{tabular}{lp{0.57\linewidth}}
\toprule
Hyperparameter & Value\\
\midrule
Base model                          & \textsc{Qwen3-14B}\\
Training data                       & Post-filter \textsc{SWE-Gym} (2{,}272) + disjoint \textsc{SWE-rebench} train split (1{,}200)\\
RL algorithm                        & GRPO with CRR advantage\\
Rollout dispatch microbatch size    & 64\\
Group size $G$                      & 8\\
Trajectories per rollout collection & 256\\
Rollout collection rounds           & about 94\\
Total on-policy rollout trajectories & 24{,}000\\
Max turns per trajectory $H$        & 50\\
Fork budget $k$                     & 4\\
Counterfactual samples per fork $K$ & 1 in the full-budget main comparison\\
Counterfactual temperature $T_{\text{cf}}$ & 0.7\\
Training rollout temperature          & 0.6 (evaluation is greedy)\\
Selector entropy threshold $\eta$   & 0.4\\
Selector type weight $\lambda_{\text{type}}$ & 0.6\\
Learning rate                       & $5\times 10^{-6}$\\
Warmup steps                        & 100\\
Optimizer minibatch/update steps    & 4{,}000\\
KL penalty $\beta$                  & 0.04\\
PPO clip $\epsilon$                 & 0.2\\
Discount $\gamma$                   & 1.0\\
LoRA rank / alpha                   & 128 / 256\\
Optimizer                           & AdamW ($\beta_1=0.9$, $\beta_2=0.95$)\\
Hardware                            & $8\times$NVIDIA H100 80GB\\
Framework                           & \textsc{SkyRL-Agent} + \textsc{SWE-MiniSandbox}\\
Main CRR training compute           & 321 H100-hours per seed; see Table~\ref{supp-accounting}\\
Random seeds                        & $\{11,23,47\}$\\
\bottomrule
\end{tabular}
\end{table}

%% file: sections/appendix-prompts.tex
\section{Prompt and Tool Catalog}
\label{sec:appendix-prompts}

All trained methods share the same prompting scaffold to isolate the contribution of the training procedure; the two zero-shot rows of Table~\ref{tab:main} use the \textsc{SWE-agent} and \textsc{OpenHands} scaffolds, respectively. The system prompt instructs the agent to act as a software engineer resolving a single GitHub issue through a fixed tool catalog; the user prompt contains the issue title, body, and an indication of the working directory. The system prompt is reproduced below for reference.

\begin{quote}\small
\textbf{System prompt (verbatim):}\\
You are an autonomous software engineer working in a sandboxed development environment. Your goal is to resolve the GitHub issue described by the user by editing the relevant source files and validating your changes against the project's test suite. You have access to the following tools: \texttt{view\_file}, \texttt{edit\_file}, \texttt{run\_shell}, \texttt{run\_tests}, \texttt{submit\_patch}. Each turn you must emit a single tool call as a JSON object. The episode terminates when you call \texttt{submit\_patch} or after $50$ turns. Read carefully before editing, edit minimally, and verify with the test suite before submission.
\end{quote}

\paragraph{Tool catalog.} The agent has five tools: \texttt{view\_file(path, start, end)} returns up to 200 lines of a file with line numbers; \texttt{edit\_file(path, start, end, replacement)} replaces a contiguous range with new content and returns a diff; \texttt{run\_shell(command, timeout)} executes a shell command in the repository root, with a default 60-second timeout; \texttt{run\_tests(test\_filter)} invokes the project's canonical test runner, filtered by an optional pytest-style selector; and \texttt{submit\_patch()} terminates the episode and triggers the held-out test suite. The tool schemas used by the harness follow these semantics, and stdout is truncated deterministically to keep observations bounded.

\paragraph{Decoding settings.} Greedy pass@1 evaluation uses temperature $0.0$, while training rollouts use temperature $0.6$; the counterfactual proposal $q_{\neg a_t}$ overrides the latter with $T_{\text{cf}}=0.7$ and rejection-conditions on an alternative action as described in Section~\ref{sec:method-proposal}. The multi-sample evaluation in Table~\ref{tab:sampling-main} instead uses evaluation temperature 0.6 and eight completions per instance. Tool calls are constrained by a JSON-grammar decoder that rejects malformed structured outputs at the token level; rejection statistics during training are reported in Section~\ref{sec:appendix-extra-ablations}. \textsc{Qwen3-14B} is run in non-thinking mode for both training rollouts and evaluation.

%% file: sections/appendix-determinism.tex
\section{Determinism Filter and Network Replay}
\label{sec:appendix-determinism}

CRR's unbiasedness guarantee (Proposition~\ref{prop:unbiased}) hinges on the assumption that the sandbox is deterministically forkable. Real SWE task instances violate this assumption in three ways: (i) tests that are flaky in time, ordering, or random seed; (ii) tests that depend on network state (package indexes, time-of-day APIs, mutable third-party services); and (iii) tests whose environment introduces nondeterminism through file-system races or thread scheduling. We describe the consistency filter and the two replay mechanisms used to enforce determinism on the training corpus.

\paragraph{Three-rerun consistency filter.} For every upstream candidate task instance we execute the canonical solution patch three times in independent sandboxes, with different random seeds and shuffled test orderings. A task is admitted to the training set only if the held-out test suite produces identical pass/fail outcomes on all three runs. This filter discards 6.8\% of the upstream \textsc{SWE-Gym} candidate pool and 9.4\% of the upstream \textsc{SWE-rebench} candidate pool. After filtering, the admitted corpus used by the main training runs contains 2{,}272 \textsc{SWE-Gym} tasks and 1{,}200 \textsc{SWE-rebench} training tasks, for 3{,}472 deterministic instances total. The 166 discarded \textsc{SWE-Gym} instances are dominated by tests that depend on set iteration order, hash randomization, or implicit time-zone conversion. Table~\ref{tab:discard} lists the repository composition of the \textsc{SWE-Gym} candidate pool together with the aggregate filter outcome; the per-repository split of the discarded instances is not reported.

\begin{table}[htbp]
\centering
\caption{Upstream \textsc{SWE-Gym} candidate pool by repository and the aggregate outcome of the three-rerun determinism filter. Repository counts are those of the public \textsc{SWE-Gym} release.}
\label{tab:discard}
\small
\begin{tabular}{lr}
\toprule
Repository & Candidate tasks\\
\midrule
\texttt{pandas}   & 737\\
\texttt{MONAI}    & 374\\
\texttt{moto}     & 343\\
\texttt{mypy}     & 257\\
\texttt{dvc}      & 225\\
\texttt{dask}     & 145\\
\texttt{modin}    & 107\\
\texttt{pydantic} & 83\\
\texttt{conan}    & 75\\
\texttt{hydra}    & 66\\
\texttt{bokeh}    & 26\\
\midrule
Discarded by the filter (all repositories) & 166 (6.8\%)\\
Admitted to the training corpus & 2{,}272\\
\bottomrule
\end{tabular}
\end{table}

\paragraph{Deterministic package mirror.} Network operations during patch installation are routed through a deterministic mirror that pins every package version, hash, and metadata blob at the time the task instance is admitted. Subsequent counterfactual rollouts hit the same frozen mirror, so \texttt{pip install} and equivalent operations return bit-identical artifacts even if upstream packages mutate during training.

\paragraph{HTTP replay layer.} For tasks whose test suite issues outbound HTTP requests (a small minority), we capture every request-response pair during canonical solution execution and replay them deterministically during training. The replay layer matches on request method, URL, and a normalized body digest; mismatches are surfaced as task-level errors and the instance is excluded.

\paragraph{\textsc{SWE-rebench} split hygiene.} The 1{,}200 admitted \textsc{SWE-rebench} training tasks are disjoint from the held-out \textsc{SWE-rebench} Python evaluation split. We split the candidate stream by collection month before filtering and then remove overlaps by task identifier, issue URL, pull-request URL, repository/commit window, exact patch hunk, near-duplicate patch hunk, and near-duplicate issue text. Held-out evaluation tasks are never used for rollout collection, counterfactual replay, selector tuning, or hyperparameter selection.

\paragraph{Residual nondeterminism.} Even after filtering and replay, a small residual nondeterminism survives at the level of file-system iteration order on certain glob-heavy tests; we treat such cases as additional noise on top of the binary terminal reward. The full-budget main comparison uses $K=1$ counterfactual sample per selected fork. In a separate residual-nondeterminism robustness setting, $K=2$ averaging reduces stochastic terminal-noise variance; this setting is excluded from the full-budget ledger in Table~\ref{supp-accounting}. Averaging is not a correction for systematic replay bias, and the additional single-seed injection experiments establish only the behavior observed under their stated reduced-budget protocol.

%% file: sections/appendix-extra-ablations.tex
\section{Additional Ablations}
\label{sec:appendix-extra-ablations}

\paragraph{Fork-pool concurrency.} The asynchronous fork pool runs counterfactual rollouts on a separate logical worker pool while the main rollout pool produces on-policy trajectories. We profile concurrent counterfactual workers $W\in\{1,2,4,8,16\}$. Table~\ref{tab:concurrency} reports relative throughput inside this fork-pool profiling harness, not the end-to-end on-policy-trajectory per H100-hour accounting used in Table~\ref{supp-accounting}. Throughput grows sublinearly: increasing $W$ from 8 to 16 changes the relative throughput from 4.25 to 4.46, about $4.9\%$. This local profile is distinct from the reduced-budget $W=1$ control.

\begin{table}[htbp]
\centering
\caption{Fork-pool profiling throughput as a function of counterfactual-worker concurrency $W$. Throughput is normalized to $W=1$ inside the profiling harness. These profiling results are not an end-to-end device-time ledger.}
\label{tab:concurrency}
\small
\begin{tabular}{lccccc}
\toprule
Concurrency $W$                   & 1     & 2     & 4     & 8     & 16\\
\midrule
Relative fork-pool throughput      & 1.00  & 1.75  & 2.89  & 4.25  & 4.46\\
\bottomrule
\end{tabular}
\end{table}

\paragraph{GRPO vs.\ PPO outer loop.} CRR can be inserted into either policy-gradient outer loop. Replacing GRPO with PPO yields $41.4\pm 0.7$\% for CRR-PPO on \textsc{SWE-bench Verified}, compared with $41.7\pm 0.6$\% for CRR-GRPO and $34.1\pm 1.0$\% for the outcome-only PPO baseline of this comparison. This baseline is a separate training run from the \textsc{RLEF} (outcome-only PPO) reproduction in Table~\ref{tab:main} ($33.5\pm 1.0$\%); the two values come from different runs and are not expected to coincide. The reported PPO improvement of $+7.3$ points is measured against this comparison's own baseline. These endpoint summaries are descriptive; they do not establish statistical equivalence between CRR-PPO and CRR-GRPO or an equal-compute PPO comparison. Eq.~\eqref{eq:hybrid} changes the selected-step advantages while retaining the outer-loop update structure.

\paragraph{Selector baselines.} The Random selector samples up to $k=4$ eligible non-terminal steps uniformly without replacement from the same entropy-thresholded eligibility set used by CRR. The Stride selector sorts eligible non-terminal steps by time and chooses approximately evenly spaced time-ordered quantiles of this eligible set. Both baselines use the same fork budget and counterfactual proposal as CRR; only the decision-index selection rule changes. The Tool only and Entropy only variants in Figure~\ref{fig:ablations}(b) score decision points by the tool-type term alone or by $\mathcal{H}_t$ alone ($\lambda_{\text{type}}=0$), respectively.

\paragraph{Tool-grammar rejection statistics.} The JSON-grammar decoder used during training rejects $1.4\%$ of tokens at the start of training and $0.3\%$ at convergence, with malformed tool calls falling fastest. We see no evidence that grammar enforcement biases the counterfactual proposal: the rejection rate at counterfactual-sampling temperature $T_{\text{cf}}{=}0.7$ is within $0.05$ percentage points of the on-policy rate at every checkpoint.

\paragraph{Counterfactual sample size $K$.} We compare the main setting $K{=}1$ (one counterfactual per fork) against the larger-$K$ configurations in Table~\ref{supp-tab-k-sweep}, all under the reduced-budget protocol (48 H100-hours per arm, one training seed). $K{=}2$ gains $+0.4$ points over $K{=}1$ at twice the fork count. We adopt $K{=}1$ in all main-text numbers; higher-$K$ runs are ablations and residual-nondeterminism robustness checks rather than part of the main accounting.

%% file: sections/appendix-perrepo.tex
\section{Per-Repository Comparison}
\label{sec:appendix-perrepo}

Figure~\ref{fig:perrepo} reports an additional per-repository comparison on \textsc{SWE-bench Verified}, using a separate set of experimental runs from the main comparison in Table~\ref{tab:main}. CRR achieves higher pass@1 than vanilla GRPO on all ten displayed repositories. Table~\ref{tab:perrepo-values} lists the displayed values together with the \textsc{SWE-bench Verified} task counts; the displayed repositories are the ten with the most task instances. Weighting all twelve repositories by their task counts gives per-repository means consistent with the aggregate pass@1 values in Table~\ref{tab:main}. These comparisons do not isolate trajectory length from other repository characteristics.

\begin{table}[htbp]
\centering
\caption{Per-repository pass@1 (\%) for the repositories displayed in Figure~\ref{fig:perrepo}, with \textsc{SWE-bench Verified} task counts. Values are pooled over the three seeds of these runs.}
\label{tab:perrepo-values}
\small
\begin{tabular}{lccc}
\toprule
Repository & Task count & Vanilla GRPO & CRR (ours)\\
\midrule
django          & 231 & $31.75$ & $38.24$\\
sympy           & 75  & $43.56$ & $48.00$\\
sphinx          & 44  & $43.18$ & $46.97$\\
matplotlib      & 34  & $41.18$ & $46.08$\\
scikit-learn    & 32  & $41.67$ & $44.79$\\
astropy         & 22  & $37.88$ & $42.42$\\
xarray          & 22  & $34.85$ & $39.39$\\
pytest          & 19  & $33.33$ & $36.84$\\
pylint          & 10  & $36.67$ & $40.00$\\
requests        & 8   & $37.50$ & $41.67$\\
\bottomrule
\end{tabular}
\end{table}

\begin{figure}[htbp]
\centering
\includegraphics[width=\linewidth]{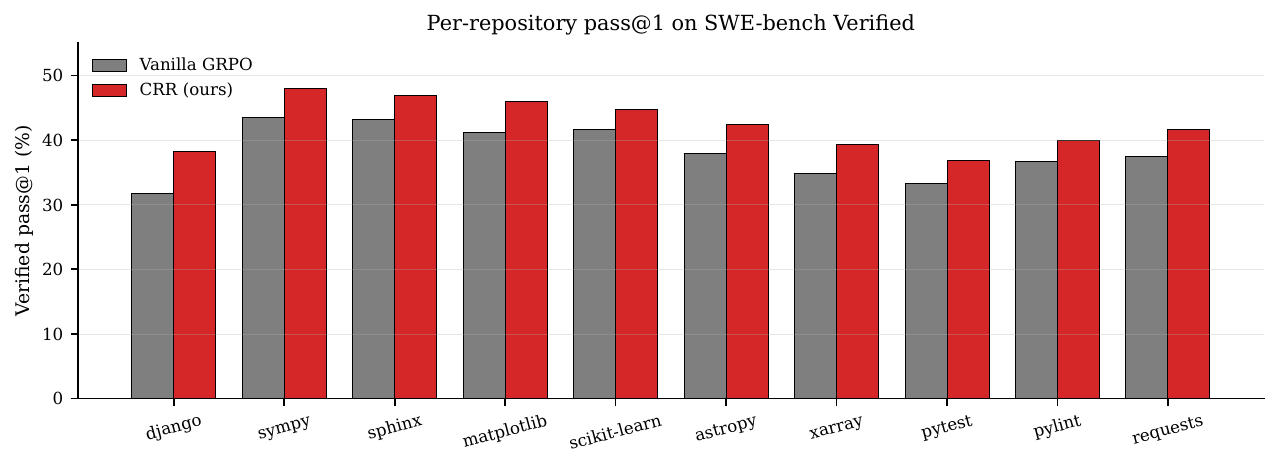}
\caption{Per-repository \textsc{SWE-bench Verified} pass@1 for vanilla GRPO and CRR under the full-budget protocol, on the ten repositories with the most task instances (Table~\ref{tab:perrepo-values}). Repository-level counts are too small per method for per-repository standard deviations.}
\label{fig:perrepo}
\end{figure}

%% file: sections/appendix-error-taxonomy.tex
\section{Error-Mode Taxonomy and Per-Mode Breakdown}
\label{sec:appendix-error-taxonomy}

We use a four-class post-hoc error taxonomy for analyzing failed trajectories and add an ``Other'' bucket for failures that the deterministic classifier cannot assign cleanly to one of the four modes. Each failed instance is labeled by a single dominant error mode through a deterministic post-hoc classifier that inspects the agent's final tool-call sequence and the diagnostic output of the held-out test suite. When an instance satisfies both the \emph{partial fix} and \emph{regression} conditions, the classifier assigns \emph{partial fix}, since the originally-failing test passes.

\begin{description}
\item[Wrong localization.] The agent edits a file or function that is not on the canonical fix path. The held-out test suite fails because the actual defect remains untouched.
\item[Partial fix.] The agent edits the correct location but misses a co-required change elsewhere (e.g.\ updates the implementation but not its caller). The held-out test suite passes the originally-failing test but fails at least one previously-passing test.
\item[Regression.] The agent's edit introduces a new failure unrelated to the issue, typically by an aggressive global rewrite or by a misfired \texttt{sed}-style transformation.
\item[Harness subversion.] The agent modifies the test harness, conftest fixtures, or assertion expectations rather than the implementation under test. Such a patch may produce an apparent verifier pass by changing the verifier target itself, but we count it as a failure after an audit check flags the modified harness files. We detect this class through a static check on the patched files against a canonical exclusion list.
\end{description}

\begin{table}[htbp]
\centering
\caption{Error-mode distribution among failed instances on \textsc{SWE-bench Verified}, normalized to percentages of all failures of the corresponding method. The table reports conditional shares among failures; Table~\ref{tab:errors-absolute} converts the main comparison to approximate absolute rates over all evaluated instances.}
\label{tab:errors}
\small
\begin{tabular}{lccccc}
\toprule
Method               & Wrong loc. & Partial fix & Regression & Harness subv. & Other\\
\midrule
Vanilla GRPO         & $34.2$     & $39.7$      & $19.4$     & $6.7$ & $0.0$\\
\textsc{SWE-Trace}   & $33.6$     & $35.1$      & $24.5$     & $6.8$ & $0.0$\\
CRR (ours)           & $39.5$     & $24.2$      & $11.3$     & $6.6$ & $18.4$\\
CRR + \textsc{SWE-Trace} & $40.8$ & $21.9$      & $9.1$      & $6.8$ & $21.4$\\
\bottomrule
\end{tabular}
\end{table}

\begin{table}[htbp]
\centering
\caption{Approximate absolute error-mode rates over all \textsc{SWE-bench Verified} instances, obtained by multiplying Table~\ref{tab:errors}'s conditional failure shares by the corresponding failure rates from Table~\ref{tab:main}.}
\label{tab:errors-absolute}
\small
\begin{tabular}{lccccc}
\toprule
Method & Wrong loc. & Partial fix & Regression & Harness subv. & Other\\
\midrule
Vanilla GRPO & $21.75$ & $25.25$ & $12.34$ & $4.26$ & $0.00$\\
CRR (ours)   & $23.03$ & $14.11$ & $6.59$  & $3.85$ & $10.73$\\
\bottomrule
\end{tabular}
\end{table}

Relative to vanilla GRPO, CRR has lower conditional shares and lower approximate absolute rates of \emph{partial fix} and \emph{regression} failures, alongside a larger residual ``Other'' bucket. This pattern does not hold uniformly across methods: \textsc{SWE-Trace}'s conditional regression share is $24.5\%$, compared with $19.4\%$ for vanilla GRPO. The error taxonomy is a descriptive post-hoc analysis and does not establish that an individual mode changed because of improved credit assignment. \emph{Harness subversion} remains present in all methods, so these results do not show that counterfactual credit alone resolves verifier exploitation. We leave the implications for future closed-loop training pipelines to future work.

%% file: sections/appendix-cases.tex
\section{Qualitative Case Studies}
\label{sec:appendix-cases}

This appendix presents three illustrative trajectories from the training run. For each case we report the on-policy trajectory, a sampled counterfactual rollout from a selected decision point, and the resulting per-step advantage. These examples illustrate how the sampled return contrast enters training; a single pair of stochastic continuations does not establish the causal contribution of an individual action.

\subsection*{Case study 1: a localization-driven win}

The agent is asked to fix a numerical-precision bug in a popular scientific library. On the on-policy trajectory the agent reads the test file at turn 3, recognizes the relevant tolerance constant, edits a single line at turn 7, runs the focused test at turn 8, and submits at turn 11; the held-out test suite passes. CRR selects turn 3 as a high-leverage decision point: at this turn the policy entropy is high (the agent could equally have read the implementation file or the documentation). The counterfactual fork samples ``view\_file(documentation.md)'' instead and continues on-policy; without the test file in context the agent fails to identify the correct tolerance and the held-out test suite fails. The resulting advantage is $\hat{A}^{\text{CRR}}_3 = 1 - 0 = 1$, credited to the early-localization decision rather than to the eventual edit.

\subsection*{Case study 2: an avoided regression}

The agent is asked to fix a serialization bug. On the on-policy trajectory the agent edits the serializer correctly at turn 12 and verifies with a focused test at turn 14, then submits at turn 16; the held-out test suite passes. CRR selects turn 12 as a high-leverage decision point. The counterfactual fork samples a more aggressive global rewrite of the serializer interface; on-policy continuation in the fork produces a patch that passes the focused test but introduces a regression in a previously-passing serialization roundtrip test. The sampled advantage is $\hat{A}^{\text{CRR}}_{12} = 1 - 0 = 1$, assigning positive credit to the realized edit relative to this sampled alternative.

\subsection*{Case study 3: a beneficial counterfactual}

The agent is asked to fix a parameter-validation bug. On the on-policy trajectory the agent edits the validator at turn 6 but introduces a typo in an exception message, causing a ``raises\_with\_message'' assertion in the held-out test suite to fail; reward is $0$. CRR selects turn 6. The counterfactual fork samples a slightly different edit at turn 6 that happens to use the correct exception message; on-policy continuation produces a passing patch with reward $1$. Advantage: $\hat{A}^{\text{CRR}}_6 = 0 - 1 = -1$, penalizing the typo-introducing edit even though the bulk of the trajectory was correct---a contrast that a single terminal outcome bit at the end of fifty turns cannot resolve.

These cases illustrate early localization, a contrast between edits with different regression outcomes, and a beneficial alternative to a failing edit. They are examples rather than an estimate of how much of the total training signal each pattern explains.

%% file: sections/appendix-cost-breakdown.tex
\section{Rollout-Dispatch Cost Breakdown}
\label{sec:appendix-cost-breakdown}

Table~\ref{tab:cost-appendix} reports the critical-path wall-clock latency of one rollout-dispatch microbatch of 64 trajectories on the $8\times$H100 node, referenced in the additional diagnostics of Section~\ref{sec:appendix-additional-discussion}. CRR multiplies the microbatch latency of vanilla GRPO by $1.71\times$, dominated by counterfactual fork/materialization, counterfactual rollout, and additional test execution; the policy update step itself is essentially unchanged. Serially extending the profiled critical path over the $375$ microbatches of a 24{,}000-trajectory run understates the end-to-end wall-clock in Table~\ref{supp-accounting} by nearly the same factor for both arms ($23.5/10.4\approx2.26$ and $40.1/17.8\approx2.25$), because the single-microbatch profile excludes sandbox startup, queueing, retries, and other node-level work that the accounting includes; the arm-level ratio of $1.71\times$ is unaffected. The largest single overhead component in our stack is counterfactual fork/materialization (36.2 seconds), followed by counterfactual rollout (28.1 seconds), suggesting that future infrastructure improvements should target both snapshot/restore materialization and asynchronous rollout scheduling.

\begin{table}[htbp]
\centering
\caption{Critical-path wall-clock latency per 64-trajectory rollout-dispatch microbatch on $8\times$H100 (seconds). The reported CRR microbatch latency is $1.71\times$ that of vanilla GRPO. This local profile does not replace the end-to-end device-time accounting in Table~\ref{supp-accounting} or the training-window comparison in Table~\ref{supp-wallclock}.}
\label{tab:cost-appendix}
\small
\begin{tabular}{lccc}
\toprule
Component & Vanilla GRPO & CRR (ours) & $\Delta$\\
\midrule
Forward rollout                   & $78.0$ & $78.0$ & $+0.0$\\
Test execution                    & $9.4$  & $9.4$  & $+0.0$\\
Counterfactual fork ($k=4$)       & $0.0$  & $36.2$ & $+36.2$\\
Counterfactual rollout            & $0.0$  & $28.1$ & $+28.1$\\
Counterfactual test execution     & $0.0$  & $6.3$  & $+6.3$\\
Policy update (GRPO step)         & $12.7$ & $12.9$ & $+0.2$\\
\midrule
Total                             & $100.1$ & $170.9$ & $+70.8$\\
\bottomrule
\end{tabular}
\end{table}

%% file: sections/appendix-compute.tex
\section{Compute Accounting and Equal-Wall-Clock Comparison}
\label{sec:appendix-compute}

The left panel of Figure~\ref{fig:sample-efficiency} measures sample efficiency in newly collected on-policy trajectories; it does not charge counterfactual replay. The consolidated device-time accounting in Table~\ref{supp-accounting} charges 188 H100-hours to vanilla GRPO, 233 to \textsc{SWE-Trace}, and 321 to CRR per full-budget run and seed. CRR's total includes 133 H100-hours of fork-pool device-time in addition to 164 for on-policy rollout and 24 for learner updates. These components partition the measured node device-time (Section~\ref{supp-compute-accounting}); overlapping their execution does not remove their device-time cost. Table~\ref{tab:accounting} summarizes the configuration and accounting conventions.

Table~\ref{supp-wallclock} provides the equal-wall-clock comparison on the same $8\times$H100 node. At 40.125 hours, CRR reaches 41.7\% on \textsc{SWE-bench Verified}, versus 36.7\% for the extended vanilla-GRPO run, a 5.0-point difference in three-seed means. This comparison includes dispatch, tests, forks, waits, retries, optimization, and checkpoints; one-time loading, final evaluation, and separate ablations are excluded. The comparison establishes the observed margin over vanilla GRPO at the reported operating points, not equal-compute dominance over every baseline or on every benchmark.

\begin{table}[htbp]
\centering
\caption{Configuration and accounting conventions for the full-budget CRR run. Method-specific device-time components are reported in Table~\ref{supp-accounting}; equal-wall-clock checkpoints are reported in Table~\ref{supp-wallclock}.}
\label{tab:accounting}
\small
\setlength{\tabcolsep}{5pt}
\begin{tabular}{lp{0.24\linewidth}p{0.42\linewidth}}
\toprule
Quantity & Main value & Meaning\\
\midrule
Training corpus & 3{,}472 tasks & post-filter admitted deterministic tasks\\
\textsc{SWE-rebench} split & disjoint & train/eval separated before filtering and decontaminated\\
On-policy trajectories & 24{,}000 & newly collected policy rollouts per full-budget CRR run/seed\\
Fork budget & $k=4$ & maximum selected decision points per trajectory\\
Counterfactual samples & $K=1$ & main setting; $K=2,4$ only in ablations\\
Main counterfactual tuples & about 96{,}000 & nominal maximum $24{,}000\times k\times K$ per seed\\
Main CRR training compute & 321 H100-hours & rollout, fork-pool, and learner device-time per seed\\
CRR node wall-clock & 40.125 hours & $321/8$ on the $8\times$H100 node\\
Cost profile unit & 64-trajectory microbatch & Table~\ref{tab:cost-appendix}, not a single trajectory\\
\bottomrule
\end{tabular}
\end{table}

%% file: sections/appendix-additional-discussion.tex
\section{Additional Diagnostics and Discussion}
\label{sec:appendix-additional-discussion}

\paragraph{Limitations.}
Supplementing the limitations in Section~\ref{sec:discussion} with environment-level detail: CRR rests on deterministic forkability, which breaks down in three regimes. Flaky tests bias the counterfactual differential; a three-rerun consistency filter discards 6.8\% of the upstream \textsc{SWE-Gym} candidate pool to enforce determinism (Section~\ref{sec:appendix-determinism}). Network-dependent operations are mitigated by a deterministic package mirror and an HTTP replay layer. Residual file-system or scheduling nondeterminism is probed by the controlled injection diagnostics of Section~\ref{supp-approximate-forkability}; the main results use $K=1$. Fork/materialization and counterfactual rollout remain the dominant training-time overheads, motivating throughput-aware scheduling that more aggressively interleaves on-policy and counterfactual rollouts. The analysis also assumes binary terminal rewards and an alternative-action proposal $q_{\neg a_t}$ with support on plausible policy actions; extending to weighted or multi-objective rewards requires generalizing Eq.~\eqref{eq:crr} to weighted return functionals.

\paragraph{Error modes and cost.}
Using the error taxonomy of Section~\ref{sec:appendix-error-taxonomy}, CRR has lower conditional shares of \emph{partial fix} and \emph{regression} failures than vanilla GRPO, but also a larger residual ``Other'' category. The approximate rates over all evaluated instances are reported separately in Table~\ref{tab:errors-absolute}; conditional shares should not be read as absolute failure rates or causal evidence about credit assignment. CRR multiplies rollout-dispatch microbatch latency by $1.71\times$ (Section~\ref{sec:appendix-cost-breakdown}). With all fork overhead charged, the equal-wall-clock endpoint in Table~\ref{supp-wallclock} is 41.7\% for CRR versus 36.7\% for vanilla GRPO on \textsc{SWE-bench Verified}.

\paragraph{Model-based framing and broader impact.}
CRR can be read as a degenerate model-based RL method~\citep{sutton1991dyna} whose ``model'' is the actual environment rather than a learned approximator: it inherits access to the true transition function but pays the full per-fork cost. Because CRR depends on a verifier, it is bounded by that verifier's integrity. In Section~\ref{sec:appendix-error-taxonomy}, \emph{harness subversion} accounts for $6.6\%$ of CRR failures and $6.7\%$ of vanilla-GRPO failures. Their approximate rates over all evaluated instances are $3.85\%$ and $4.26\%$, respectively; these descriptive results do not establish a specific effect on verifier exploitation. As recommended in Section~\ref{sec:discussion}, CRR-trained agents should be paired with adversarial test-suite hardening~\citep{yu2026sweabs} and execution-free corroboration such as \textsc{SWE-RM}~\citep{shum2025swerm}. AI-assisted cybersecurity reviews also emphasize data protection and human oversight~\citep{weng2024leveraging}. Every counterfactual rollout produces an auditable execution trace that can support examination of the resulting process signal.

\paragraph{Semantic understanding and the conditioning context.}
The policy's representation and interpretation of a task influence the decisions available for credit assignment. Domain-adaptive pretraining for multimodal governance uses task-specific perception, question answering, and reasoning supervision~\citep{wang-etal-2025-reasoning-enhanced}. Complementary work examines prompt structure~\citep{zhang2025does}, question calibration and multi-hop temporal reasoning~\citep{xue2024question}, fine-grained semantic conflicts~\citep{xue2023dual}, and structured contrastive semantic matching~\citep{xue2025structcoh}. For visual settings, complementary upstream challenges include foreground--background separation in open-vocabulary segmentation~\citep{yao2026disa} and spatiotemporal registration between RGB and event observations~\citep{yao2025learning}. Related rendering and reconstruction work couples learned representations with physical image formation, including geometry-guided ray tracing for dynamic specular scenes~\citep{tao2026reflflow}, reusable Gaussian photon fields for light transport~\citep{tao2026particlesfieldsreframingphoton}, and attenuation-aware X-ray reconstruction~\citep{Wang2025XFieldAP}. These approaches address upstream understanding and representation; they provide possible complements to learning from the execution outcomes of selected actions.

\paragraph{Feedback and supervision-budget allocation.}
Evidence-conditioned critique and iterative revision support retrieval-augmented reasoning~\citep{weiretrieval}, while disclosure policies learn when reasoning should continue before emitting supported content~\citep{wei2026think}. Selection of additional supervision can also exploit error neighborhoods in representation space~\citep{10.1145/3711896.3737195}. These mechanisms suggest possible alternatives or complements to CRR's entropy-based allocation of counterfactual rollouts. Their usefulness for a fork selector would require evaluation on the same SWE tasks and compute budget.

\paragraph{Broader agent environments and evaluation coverage.}
Agent workflows also appear in automated scientific discovery~\citep{wei2025ai,kong2026aiautoresearch} and multi-agent poster generation~\citep{zhang2025postergen}. Embodied action generation can use intermediate vision-language representations routed to an action decoder~\citep{zhang2026liralocalcrosslayerinformation}. In decentralized multi-agent control, local observations and communication constraints shape the information available to a policy~\citep{fan2026communicationawaremultiagentreinforcementlearning}. Extending executable contrasts to these settings would depend on reliable state restoration and an appropriate verifier. Cross-lingual reasoning training and culturally grounded multimodal evaluation~\citep{zheng2025ccl,zheng2026mmac} additionally identify coverage dimensions beyond the current SWE experiments. Remote-sensing video question answering further tests scene-centric spatiotemporal reasoning~\citep{yao2026revascenecentricdatasetrepetition}. Evaluation also extends beyond aggregate answer scores: trace-level acknowledgment of injected biasing content~\citep{sun2026accuracymeasuringbiasacknowledgment} and configuration-dependent model rankings on alignment benchmarks~\citep{li2026safetyreproconfigurationconditionalrankinstability} address distinct aspects of evaluation reliability.

\paragraph{Task objectives, planning, and deployment cost.}
Reward design for UAV hovering combines position holding, control effort, and disturbance response~\citep{11331433}; risk-aware grid planning incorporates obstacle proximity and localization uncertainty into search costs~\citep{11549609}. These illustrate different uses of task feedback: specifying control objectives or ranking paths. CRR constructs credit for selected actions under a given terminal verifier. Separating offline model-assisted computation from deployment-time prediction is also explored in few-shot tabular learning~\citep{yang2026forestllm}. This provides a cross-domain perspective on reporting where additional computation occurs; CRR's reported costs explicitly charge replay during training.

%% file: sections/appendix-dataset.tex
\section{Counterfactual Trajectory Dataset}
\label{sec:appendix-dataset}

A side artifact of training is a dataset of counterfactual tuples that records, for every selected decision point, the on-policy and counterfactual returns together with the replay metadata needed to audit the fork. The main run uses 24{,}000 on-policy rollout trajectories, fork budget $k=4$, and $K=1$, giving a nominal maximum of $24{,}000\times k\times K = 96{,}000$ counterfactual tuples per seed and up to 288{,}000 tuples across the three main seeds. Across all seeds, ablations, and diagnostic runs, the dataset to be released contains approximately 1.4 million tuples spanning 3{,}472 task instances. The three main seeds contribute at most 288{,}000 tuples at the nominal fork maximum; the sensitivity program runs under the reduced-budget protocol (48 H100-hours per arm), in which a CRR arm at native fork cost completes roughly 6{,}000 trajectories (Table~\ref{supp-tab-overhead}) rather than 24{,}000, so its nominal fork contribution is commensurately smaller per run. Each tuple contains the task id, policy seed/version, selected index $t$, sandbox snapshot identifier, realized action $a_t$, counterfactual action $\tilde{a}_t$, on-policy and counterfactual action suffixes sufficient for replay, terminal returns $R(\tau)$ and $R(\tilde{\tau}_t)$, proposal log-probability, and selector score. We expect this dataset to be useful as direct supervision for downstream PRM training, for verifier evaluation, and for studying the structure of credit assignment in long-horizon code agent trajectories. The dataset will be released under a license compatible with the \textsc{SWE-Gym} release; tuples derived from \textsc{SWE-rebench} tasks additionally follow the licensing terms of the \textsc{SWE-rebench} release.

%% file: sections/appendix-reproducibility.tex
\section{Reproducibility Details}
\label{sec:appendix-reproducibility}

\paragraph{Code.} The reference implementation is built on top of \textsc{SkyRL-Agent}~\citep{cao2025skyrl} and \textsc{SWE-MiniSandbox}~\citep{yuan2026minisandbox} with no proprietary dependencies. The CRR-specific additions amount to fewer than 700 lines of Python.

\paragraph{Data.} All training tasks are drawn from the public \textsc{SWE-Gym} and \textsc{SWE-rebench} releases. The main training corpus contains 3{,}472 post-filter deterministic tasks, and the \textsc{SWE-rebench} training split is disjoint from the held-out \textsc{SWE-rebench} Python evaluation split by the rules in Section~\ref{sec:appendix-determinism}. The flaky-test filter, the package mirror configuration, and the recorded HTTP responses will be released with the code.

\paragraph{Compute.} The full-budget device-time totals are 188 H100-hours for vanilla GRPO, 233 for \textsc{SWE-Trace}, and 321 for CRR per seed (Table~\ref{supp-accounting}). CRR's 321 H100-hours correspond to 40.125 hours on the $8\times$H100 node. Table~\ref{supp-wallclock} reports separate equal-wall-clock checkpoints, including all fork overhead. The reduced-budget controls use 48 H100-hours per arm. These training-window costs exclude one-time loading, final evaluation, and separate ablations; they are not an aggregate cost for reproducing every experiment.

\paragraph{Random seeds and uncertainty.} The full-budget main experiments use seeds $\{11, 23, 47\}$ at the level of policy initialization, the GRPO group sampler, and the counterfactual proposal $q_{\neg a_t}$. Reported main-table standard deviations describe variation across these three seeds; for the untrained zero-shot rows of Table~\ref{tab:main}, they describe variation across three repeated greedy evaluations, which reflects inference and sandbox nondeterminism. The equal-wall-clock curves and matched Search-Select control also report three-seed means; the latter includes seed standard deviations. The reduced-budget fork-overhead, replay-noise, validity-filtering, and fixed-budget $K$ sweeps are single-seed diagnostics and do not estimate seed variability. Small differences in these sweeps do not establish significance or equivalence.

\paragraph{Limitations of reproduction.} Because CRR depends on environment determinism, exact bit-level reproduction of pass-rates may differ across sandbox backends and operating-system kernel versions. The controlled nondeterminism diagnostics characterize only the stated injection protocol; they do not guarantee robustness to arbitrary systematic changes in environment behavior.

%% file: sections/appendix-extended.tex
\section{Extended Evaluation and Resource Accounting}
\label{sec:supplementary-evaluation}

This section provides the resource ledger, matched-budget controls, and diagnostic experiments that complement the main evaluation. We distinguish full-budget training runs from reduced-budget comparisons. All device-time is expressed in H100-hours; wall-clock time is elapsed time on the specified hardware. Equal device-time does not imply equal numbers of completed trajectories, because replay and proposal costs differ between configurations.

\paragraph{Evaluation protocols.}
The full-budget experiments use the final checkpoints from the main comparison. The reduced-budget protocol allocates 48 H100-hours to each arm, holds the scaffold and hyperparameters fixed except for the tested factor, and evaluates the final checkpoint greedily on all 500 \textsc{SWE-bench Verified} instances. Reduced-budget sensitivity experiments use one training seed (the middle seed, 23, of the three-seed protocol in Table~\ref{tab:hparams}) unless stated otherwise. The matched Search-Select comparison instead uses three seeds per arm. The single-seed sensitivity experiments share one GRPO reference run (33.4\%), whereas Tables~\ref{tab:search-select} and~\ref{supp-tab-search-select} report the three-seed GRPO mean ($33.3\pm0.7$\%). Means and seed standard deviations are reported only where available.

\subsection{Complete device-time accounting}
\label{supp-compute-accounting}

Table~\ref{supp-accounting} separates on-policy generation, counterfactual workers, auxiliary scoring, and learner updates. The categories partition each run's measured node device-time: on-policy rollout and learner updates are charged at the same per-trajectory rate for all three methods, and each method-specific entry (the counterfactual fork pool for CRR, \textsc{SWE-Trace} scoring for \textsc{SWE-Trace}) is the remainder of that run's measured node time. Fork workers are logical roles on the same 8-H100 node, rather than additional uncharged hardware. The counterfactual pool contributes 133 H100-hours, bringing CRR's full-budget total to 321 H100-hours. This complete total is the reference for interpreting the full-budget comparison.

\begin{table}[htbp]
\centering
\small
\setlength{\tabcolsep}{5pt}
\caption{Full-budget resource accounting and \textsc{SWE-bench Verified} pass@1. Performance is the mean $\pm$ seed standard deviation over three training seeds. Device-time includes all listed worker roles. Wall-clock durations in this table are rounded to one decimal place; the exact CRR endpoint used in the equal-wall-clock comparison is 40.125 hours.}
\label{supp-accounting}
\begin{tabular}{lrrr}
\toprule
Accounting item & Vanilla GRPO & \textsc{SWE-Trace} & CRR \\
\midrule
On-policy rollout (H100-hours) & 164 & 164 & 164 \\
Counterfactual fork pool (H100-hours) & 0 & 0 & 133 \\
\textsc{SWE-Trace} scoring / auxiliary work (H100-hours) & 0 & 45 & 0 \\
Learner updates (H100-hours) & 24 & 24 & 24 \\
\midrule
Total device-time (H100-hours) & 188 & 233 & 321 \\
Wall-clock on one $8\times$H100 node (hours) & 23.5 & 29.1 & 40.1 \\
Verified pass@1 (\%) & $36.4\pm0.8$ & $39.1\pm0.7$ & $41.7\pm0.6$ \\
\bottomrule
\end{tabular}
\end{table}

The node is saturated: rollout, fork, and learner work share the eight GPUs, so overlap changes scheduling but not total device-time, and elapsed time equals device-time divided by eight ($188=8\times23.5$ and $321=8\times40.125$). Logical fork-slot utilization was 78\%. In elapsed time, 58\% of fork time overlapped on-policy generation, but this overlap is not subtracted from device-time. Consequently, the sample-efficiency comparison and the resource-efficiency comparison answer different questions: the former counts trajectories, while the latter must charge the additional counterfactual work.

\subsection{Equal-wall-clock learning curves and cumulative ledger}
\label{supp-equal-time}

The equal-wall-clock comparison trains both methods on the same $8\times$H100 node and extends GRPO to CRR's full training-window duration. Each column in Table~\ref{supp-wallclock} corresponds to an observed checkpoint from one continuing run per seed, rather than an interpolation between independently trained endpoints. Both arms use 100 optimizer updates of warmup to a learning rate of $5\times10^{-6}$, followed by a constant learning rate.

\begin{table}[htbp]
\centering
\small
\setlength{\tabcolsep}{4pt}
\caption{Full-budget equal-wall-clock comparison on the same $8\times$H100 node. pass@1 is the three-seed mean on \textsc{SWE-bench Verified}; trajectory, update, and token rows give the cumulative ledger at the stated checkpoints. Token counts are in millions. All four checkpoints are observed.}
\label{supp-wallclock}
\begin{tabular}{lrrrr}
\toprule
Fraction of CRR training window & 25\% & 50\% & 75\% & 100\% \\
Wall-clock (hours) & 10.031 & 20.063 & 30.094 & 40.125 \\
\midrule
GRPO pass@1 (\%) & 34.6 & 36.1 & 36.5 & 36.7 \\
GRPO completed trajectories & 11,648 & 20,992 & 29,632 & 37,824 \\
GRPO optimizer updates & 1,941 & 3,499 & 4,939 & 6,304 \\
GRPO generated tokens (M) & 28.7 & 57.4 & 86.1 & 114.8 \\
GRPO effective tokens (M) & 23.5 & 47.1 & 70.6 & 94.1 \\
\midrule
CRR pass@1 (\%) & 37.2 & 40.1 & 40.9 & 41.7 \\
CRR completed on-policy trajectories & 8,704 & 14,464 & 19,456 & 24,000 \\
CRR optimizer updates & 1,451 & 2,411 & 3,243 & 4,000 \\
CRR on-policy generated tokens (M) & 20.9 & 39.8 & 58.4 & 76.8 \\
CRR all-generated tokens (M) & 32.4 & 61.7 & 90.5 & 119.0 \\
CRR effective tokens (M) & 17.1 & 32.6 & 47.9 & 63.0 \\
\bottomrule
\end{tabular}
\end{table}

Generated tokens include counterfactual suffixes for CRR. The $42.2$M counterfactual suffix tokens ($119.0-76.8$) over the nominal maximum of $96{,}000$ forks per seed ($24{,}000\times k\times K$; Section~\ref{sec:appendix-dataset}) imply a mean suffix of $440$ or more generated tokens, versus roughly $3{,}200$ for a full on-policy trajectory; with up to four forks per trajectory, counterfactual suffixes account for about 35\% of CRR's generated tokens ($42.2$M of $119.0$M). Effective tokens are unique, loss-masked on-policy tokens. The training window begins at dispatch and ends at the last committed update. It includes tests, forks, waits, retries, optimization, and checkpointing; one-time loading, final evaluation, and separate ablation runs are excluded from both arms. At 40.125 hours, CRR reaches 41.7\% and extended GRPO reaches 36.7\%, an observed difference of 5.0 percentage points with the fork overhead fully charged. The larger number of GRPO trajectories and optimizer updates at this endpoint is part of the equal-wall-clock comparison, not a quantity held fixed.

\subsection{Sensitivity to fork overhead and concurrency}
\label{supp-overhead}

Table~\ref{supp-tab-overhead} varies the cost multiplier $m$ for the measured snapshot, restore, and replay component. Rollout and learner costs remain fixed, and $m=1$ denotes the native cost. Each arm receives the same 48 H100-hour budget, so increasing $m$ reduces the number of on-policy trajectories completed before the cutoff.

\begin{table}[htbp]
\centering
\small
\caption{Fork-overhead sensitivity under the reduced-budget protocol: 48 H100-hours per arm, one training seed, and greedy pass@1 on all 500 Verified instances. The GRPO reference is the same 33.4\% run for each cost multiplier.}
\label{supp-tab-overhead}
\begin{tabular}{lrrr}
\toprule
Fork-cost multiplier & $m=1$ & $m=2$ & $m=4$ \\
\midrule
CRR pass@1 (\%) & 36.2 & 34.8 & 33.2 \\
Vanilla GRPO pass@1 (\%) & 33.4 & 33.4 & 33.4 \\
CRR minus GRPO (percentage points) & $+2.8$ & $+1.4$ & $-0.2$ \\
CRR completed on-policy trajectories & $\approx6{,}000$ & $\approx4{,}312$ & $\approx2{,}760$ \\
\bottomrule
\end{tabular}
\end{table}

The reported completion counts are described by the cost relation
\begin{equation}
 n_m = \frac{48}{0.004868+0.003132m},
 \label{supp-eq-overhead-counts}
\end{equation}
where the numerator is the device-time budget in H100-hours and the denominator is the per-trajectory cost fitted on these reduced-budget runs. Counts reflect asynchronous completions at the time cutoff. Because the reduced-budget runs cover only the early, shorter-trajectory phase of training, the fitted per-trajectory cost is lower than the full-budget average implied by Table~\ref{supp-accounting}. The observed CRR--GRPO margin changes sign between $m=2$ and $m=4$. This interval characterizes the measured single-seed operating points; it is not a universal break-even guarantee for other workloads or hardware.

\paragraph{Serial fork-worker control.}
An additional reduced-budget control with one fork worker ($W=1$) reaches 34.8\% compared with the 33.4\% GRPO reference. Thus a gain is observed even without a large concurrent fork pool at this operating point. This result and the full-budget timing comparison do not isolate severe variation in trajectory lengths. Long suffixes, expensive restoration, or highly skewed trajectory lengths may reduce throughput and alter the useful operating range. Concurrency, scheduler behavior, rollout caps, and the costs charged to each worker role must therefore accompany any transfer of these results to another system.

\subsection{Credit assignment versus selecting a branch}
\label{supp-search-select}

The Search-Select control (Table~\ref{supp-tab-search-select}) isolates two uses of the same executable forks. Search-Select uses branch returns to select a training sequence, whereas CRR retains the realized trajectory and uses the return differential to assign credit to its actions. Both fork-using arms share the selector, number of fork executions per realized trajectory, branch-scoring verifier, rollout horizon, group size, device-time budget, and clipped learner.

\begin{table}[htbp]
\centering
\small
\caption{Matched Search-Select comparison. All arms use 48 H100-hours per seed and greedy evaluation on all 500 Verified instances. Values are the mean $\pm$ seed standard deviation over three training seeds. CRR and Search-Select use the same fork count per realized trajectory; GRPO has no forks.}
\label{supp-tab-search-select}
\begin{tabular}{llr}
\toprule
Method & Use of forks & pass@1 (\%) \\
\midrule
Vanilla GRPO & None & $33.3\pm0.7$ \\
Search-Select-GRPO & Select the training sequence & $34.4\pm0.8$ \\
CRR & Credit assignment on realized actions & $36.2\pm0.6$ \\
\bottomrule
\end{tabular}
\end{table}

\paragraph{Search-Select protocol.}
Each realized trajectory uses the same selector and $k=4$ fork executions as CRR, with identical rollout caps and verifier endpoints. Branches are scored with the terminal verifier return $R(\cdot)$; no process reward model or separate judge is used. The verifier-best fork-derived sequence replaces the realized training sequence only if its return strictly exceeds that of the realized sequence. A tie retains the realized sequence. The resulting $G=8$ sequences form the GRPO group, with ordinary group-relative advantages computed from their returns. Unselected branches are discarded and enter neither the group nor its advantage calculation. In particular, Search-Select does not reuse CRR's counterfactual return differential.

The behavior-policy log-probabilities of each generated sequence are retained, and current-policy log-probabilities are recomputed during optimization to form the clipped current-to-behavior ratio. This ratio accounts for policy drift; it does not remove the selection of the training distribution. The control deliberately implements an expert-iteration-style distribution in which a better branch may replace the realized sequence. The replacement rate is 12\%, and PPO clipping is activated for 26\% of the decision points in the replacement sequences.

CRR exceeds Search-Select by 1.8 percentage points on average. These results support the benefit of the credit-assignment use of fork data under the matched protocols. They do not imply that selective branching or entropy-guided rollout allocation is itself new, nor do they establish dominance over all forms of inference-time search.

\subsection{Multiple-sample evaluation}
\label{supp-multiple-sample-evaluation}

Table~\ref{supp-tab-passk} evaluates the full-budget final checkpoints using eight temperature-sampled trajectories per instance. Besides greedy and sampled pass@1, it reports pass@2, pass@4, pass@8, and all-8. The last metric is the fraction of instances for which all eight samples succeed, and measures consistency across sampled trajectories.

\begin{table}[htbp]
\centering
\small
\setlength{\tabcolsep}{4pt}
\caption{Multiple-sample evaluation on \textsc{SWE-bench Verified} using full-budget final checkpoints and three training seeds. Each instance receives eight independent samples at temperature 0.6. All values are percentages; the greedy column is evaluated separately from the sampled columns.}
\label{supp-tab-passk}
\begin{tabular}{lrrrrrr}
\toprule
Method & \shortstack{pass@1\\greedy} & \shortstack{pass@1\\sampled} & pass@2 & pass@4 & pass@8 & all-8 \\
\midrule
Vanilla GRPO & 36.4 & 35.9 & 43.4 & 48.8 & 53.8 & 18.7 \\
\textsc{SWE-Trace} & 39.1 & 38.6 & 45.7 & 50.8 & 55.6 & 21.2 \\
CRR & 41.7 & 41.2 & 48.1 & 53.1 & 57.7 & 23.9 \\
\bottomrule
\end{tabular}
\end{table}

The CRR advantage over GRPO persists at pass@8, with a difference of 3.9 percentage points. all-8 improves by 5.2 points, indicating greater consistency across samples. These are trajectory-level evaluation metrics and do not directly measure the quality or causal effect of individual decisions.

\subsection{Counterfactual sample count and allocation across steps}
\label{supp-counterfactual-count}

Let $K$ denote the number of counterfactual samples per selected step and $k$ the number of selected decision points. Table~\ref{supp-tab-k-sweep} compares allocating four forks to more decision points with allocating them to more alternatives at fewer points. All arms retain the 48 H100-hour device-time budget, so completed trajectory counts need not match. The matched-fork rows hold $kK$ fixed per realized trajectory, rather than asserting that the total number of forks completed during training is identical.

\begin{table}[htbp]
\centering
\small
\caption{Counterfactual-sample sensitivity under the reduced-budget protocol: 48 H100-hours per arm, one training seed, and greedy evaluation on all 500 Verified instances. Advantage SNR uses the main evaluation's diagnostic and is expressed relative to GRPO.}
\label{supp-tab-k-sweep}
\begin{tabular}{lrrr}
\toprule
Configuration & \shortstack{Forks per\\trajectory} & pass@1 (\%) & \shortstack{Relative\\advantage SNR} \\
\midrule
$K=1,\ k=4$ (default) & 4 & 36.2 & $2.1\times$ \\
$K=2,\ k=4$ (twice as many forks) & 8 & 36.6 & $2.55\times$ \\
$K=2,\ k=2$ (matched forks) & 4 & 35.8 & $2.35\times$ \\
$K=4,\ k=1$ (matched forks) & 4 & 35.0 & $2.8\times$ \\
\bottomrule
\end{tabular}
\end{table}

Within this sweep, allocating forks across more decision points performs at least as well as concentrating them on more alternatives at fewer points. Increasing $K$ from one to two at fixed $k=4$ raises the measured SNR by 21\% ($2.55/2.1$), below the $\sqrt{2}\approx1.41$ factor expected if counterfactual sampling were the only source of per-step variance, but doubles the fork count per trajectory. The observed 0.4-point pass@1 difference is not resolved by this single-seed comparison. For context, the full-budget anchors in Table~\ref{supp-accounting} have seed standard deviations of 0.6--0.8 points; those standard deviations are not uncertainty estimates for the reduced-budget sweep itself.

\subsection{Controlled departures from deterministic replay}
\label{supp-approximate-forkability}

The executable substrate and three-rerun consistency filter are shared by the main-comparison methods. To probe imperfect replay, a fraction $\nu$ of replayed steps receives an output resampled from the recorded rerun distribution. Table~\ref{supp-tab-noise} compares a single counterfactual with an averaged contrast using $K=2$ counterfactual samples. These perturbations model the specified output-resampling mechanism, rather than every possible source of environment nondeterminism.

\begin{table}[htbp]
\centering
\small
\caption{Controlled replay-noise sensitivity under the reduced-budget protocol: 48 H100-hours per arm, one training seed, and greedy pass@1 on all 500 Verified instances. The mean-advantage-shift row is a diagnostic of the injected perturbation, not a pass-rate.}
\label{supp-tab-noise}
\begin{tabular}{lrrr}
\toprule
Injected nondeterminism $\nu$ & 0\% & 5\% & 10\% \\
\midrule
CRR, $K=1$: pass@1 (\%) & 36.2 & 35.6 & 34.8 \\
CRR, $K=2$: pass@1 (\%) & 36.6 & 36.0 & 35.4 \\
Vanilla GRPO reference: pass@1 (\%) & 33.4 & 33.4 & 33.4 \\
Mean advantage shift & 0 & 0.01 & 0.02 \\
\bottomrule
\end{tabular}
\end{table}

Performance declines gradually at the measured noise levels, and $K=2$ recovers 0.4--0.6 percentage points relative to $K=1$ under the same perturbation level. The noise-free $K{=}2$ entry is the same single-seed run as the $K{=}2$, $k{=}4$ arm of Table~\ref{supp-tab-k-sweep}. The zero-noise advantage shift is zero by construction; the other observed shifts are 0.01--0.02. These observations are consistent with the specified approximately symmetric injection mainly adding variance, but do not establish unbiasedness under arbitrary replay noise. Averaging samples cannot generally remove a systematic shift in the counterfactual return distribution. Reliable state restoration and a clearly characterized replay-noise model remain part of CRR's operating conditions.

\subsection{Counterfactual quality and validity filtering}
\label{supp-proposal-quality}

An immediately invalid action can produce an easy contrast, but need not force its branch's terminal return to zero: subsequent actions may recover. We classify 20,000 logged forks into immediately invalid actions, valid actions whose branches fail, and valid actions whose branches succeed. Table~\ref{supp-tab-proposal-composition} reports each class frequency $p_c$ and class-mean contrast $\bar\Delta_c$, where $\Delta=R(\tau)-R(\tilde\tau)$.

\begin{table}[htbp]
\centering
\small
\setlength{\tabcolsep}{5pt}
\caption{Composition of 20,000 logged counterfactual forks. This is a log diagnostic, not a separate 500-instance evaluation or a multi-seed performance comparison. The final column is the normalized class-mean plug-in quantity $p_c|\bar\Delta_c|/\sum_j p_j|\bar\Delta_j|$, expressed as a percentage.}
\label{supp-tab-proposal-composition}
\begin{tabular}{lrrr}
\toprule
Alternative class & \shortstack{Fork\\frequency (\%)} & \shortstack{Mean\\contrast} & \shortstack{Plug-in\\share (\%)} \\
\midrule
Immediately invalid & 12 & $+0.15$ & 8 \\
Valid, branch fails & 46 & $+0.40$ & 83 \\
Valid, branch succeeds & 42 & $-0.05$ & 9 \\
\bottomrule
\end{tabular}
\end{table}

Valid but failing alternatives account for 83\% of this class-mean plug-in quantity, compared with 8\% for immediately invalid alternatives. The quantity summarizes class frequencies and class means. It is not the fraction of total gradient magnitude, nor the fraction of the sum of absolute per-sample contrasts, because taking the absolute value after averaging within a class can cancel contrasts of different signs.

For a more direct check, the validity-filtered proposal resamples an alternative up to three times until a syntax/schema check passes; if none passes, the step is skipped. Table~\ref{supp-tab-validity-filter} reports the resulting performance.

\begin{table}[htbp]
\centering
\small
\caption{Validity-filtered proposal under the reduced-budget protocol: 48 H100-hours per arm, one training seed, and greedy evaluation on all 500 Verified instances.}
\label{supp-tab-validity-filter}
\begin{tabular}{lr}
\toprule
Proposal & pass@1 (\%) \\
\midrule
Unfiltered counterfactual proposal & 36.2 \\
Validity-filtered counterfactual proposal & 36.4 \\
\bottomrule
\end{tabular}
\end{table}

The 0.2-point difference does not establish equivalence. It provides no evidence in this experiment that CRR's measured gain requires immediately invalid alternatives.

\paragraph{Proposal freshness and multiple samples.}
The proposal uses the current policy with temperature adjustment and is sampled in the same rollout batch as the realized trajectory, rather than from a separately trained proposal model. This describes the measured pipeline and is not a general guarantee for more asynchronous deployments. When multiple counterfactual samples disagree in their terminal returns, the contrast uses their average; the resource and SNR trade-off is examined in Table~\ref{supp-tab-k-sweep}.

\subsection{Decision-relevance diagnostics for the selector}
\label{supp-selector-diagnostics}

We evaluate selector alignment (Table~\ref{supp-tab-selector-diagnostics}) using two complementary data collections. The gold-patch-file first-touch statistic uses 2,000 logged training trajectories. The two advantage-magnitude statistics use a separate 300-trajectory probe in which counterfactuals are additionally evaluated at random, stride, and entropy-only step sets; the probe uses $K=4$ counterfactual samples per selected step, averaged as in Section~\ref{supp-proposal-quality}, so its advantage magnitudes are smoothed relative to the $K=1$ training signal. This probe allows the selectors to be compared against the same counterfactual-advantage distribution.

\begin{table}[htbp]
\centering
\small
\setlength{\tabcolsep}{4pt}
\caption{Selector alignment with decision-relevance proxies. The first numeric column uses 2,000 logged training trajectories; the remaining columns use the shared 300-trajectory probe. These are diagnostics of training trajectories, not pass@1 measurements on the 500-instance test set.}
\label{supp-tab-selector-diagnostics}
\begin{tabular}{lrrr}
\toprule
Selector & \shortstack{Hits gold-patch file\\on first touch (\%)} & \shortstack{Selected steps in\\top-decile $|\hat A^{\text{CRR}}_t|$ (\%)} & \shortstack{Mean $|\hat A^{\text{CRR}}_t|$ at\\selected steps} \\
\midrule
Random & 13 & 10 & 0.11 \\
Stride & 15 & 11 & 0.11 \\
Entropy only & 30 & 24 & 0.18 \\
Entropy + tool type & 41 & 28 & 0.21 \\
\bottomrule
\end{tabular}
\end{table}

Relative to random selection, entropy plus tool type has approximately $3.2\times$ the gold-patch-file first-touch hit rate and $2.8\times$ the top-decile advantage enrichment. These results indicate that the selector allocates evaluations to steps enriched for patch-relevant edits and larger executable return contrasts. Gold-patch overlap and sampled contrast magnitude are proxies; neither establishes an individual step's causal role in the final outcome.

%% file: sections/checklist.tex
\section*{NeurIPS Paper Checklist}
\label{sec:checklist}

\begin{enumerate}

\item {\bf Claims}\\
Question: Do the main claims made in the abstract and introduction accurately reflect the paper's contributions and scope?\\
Answer: \answerYes{}\\
Justification: The abstract and Section~\ref{sec:intro} describe executable counterfactual credit assignment, the fixed-index CRR contrast estimator, empirical comparisons, and the operating limits of counterfactual replay. The estimator and its assumptions are stated in Section~\ref{sec:method} (Propositions~\ref{prop:unbiased} and~\ref{prop:variance}), the empirical comparisons appear in Table~\ref{tab:main}, Table~\ref{tab:composition}, and Figures~\ref{fig:sample-efficiency}--\ref{fig:ablations}, and the dataset is described in Section~\ref{sec:appendix-dataset}. The equal-wall-clock comparison in Table~\ref{supp-wallclock} is restricted to vanilla GRPO on \textsc{SWE-bench Verified}; the broader main-table comparison is not described as uniformly compute-matched. We do not claim universality across RL domains.

\item {\bf Limitations}\\
Question: Does the paper discuss the limitations of the work performed by the authors?\\
Answer: \answerYes{}\\
Justification: Section~\ref{sec:discussion} discusses limitations, including expensive restoration, residual policy-sampling variance after the fork, selector-induced weighting, and verifier dependence; flaky tests and network-dependent operations are discussed in Section~\ref{sec:appendix-additional-discussion}, and the determinism filter that addresses them is described in Section~\ref{sec:appendix-determinism}. Section~\ref{sec:appendix-determinism} reports the empirical discard rate of the upstream determinism filter (6.8\% on \textsc{SWE-Gym}, 9.4\% on \textsc{SWE-rebench}). Section~\ref{sec:appendix-error-taxonomy} reports that harness subversion remains present in the evaluated methods; these observations do not show that counterfactual credit resolves verifier exploitation.

\item {\bf Theory assumptions and proofs}\\
Question: For each theoretical result, does the paper provide the full set of assumptions and a complete (and correct) proof?\\
Answer: \answerYes{}\\
Justification: Proposition~\ref{prop:unbiased} is stated for fixed selected indices with explicit assumptions (sandbox determinism and forkability at horizon $H-t$, an action-augmented history, and an alternative-action proposal $q_{\neg a_t}$). Proposition~\ref{prop:variance} is a variance decomposition with an explicit sufficient condition, not a universal variance-reduction guarantee. Complete proofs are given in Section~\ref{sec:appendix-proofs}. Section~\ref{sec:appendix-snr} reports prefix-bucket covariance diagnostics that support the empirical regime in which the condition is plausible.

\item {\bf Experimental result reproducibility}\\
Question: Does the paper fully disclose all the information needed to reproduce the main experimental results of the paper to the extent that it affects the main claims and/or conclusions of the paper (regardless of whether the code and data are provided or not)?\\
Answer: \answerYes{}\\
Justification: The main full-budget CRR hyperparameters are listed in Table~\ref{tab:hparams}, the system prompt is reproduced verbatim and the tool catalog is described in Section~\ref{sec:appendix-prompts}, the determinism filter, \textsc{SWE-rebench} split hygiene, HTTP replay layer, selector entropy approximation, and tool-type mapping are described in Sections~\ref{sec:method-selector} and~\ref{sec:appendix-determinism}, and Algorithm~\ref{alg:crr} specifies the per-iteration training loop. The main experiments use seeds $\{11, 23, 47\}$. The reduced-budget controls state their 48 H100-hour limit and distinguish three-seed comparisons from single-seed diagnostics; Section~\ref{sec:appendix-reproducibility} summarizes these protocols.

\item {\bf Open access to data and code}\\
Question: Does the paper provide open access to the data and code, with sufficient instructions to faithfully reproduce the main experimental results, as described in supplemental material?\\
Answer: \answerNo{}\\
Justification: We will release the CRR implementation (fewer than 700 lines of Python on top of \textsc{SkyRL-Agent} and \textsc{SWE-MiniSandbox}), the deterministic package mirror configuration, the recorded HTTP responses, and the counterfactual trajectory dataset of approximately 1.4M tuples (Section~\ref{sec:appendix-dataset}). Training data is drawn from public \textsc{SWE-Gym} and \textsc{SWE-rebench} releases.

\item {\bf Experimental setting/details}\\
Question: Does the paper specify all the training and test details (e.g., data splits, hyperparameters, how they were chosen, type of optimizer, etc.) necessary to understand the results?\\
Answer: \answerYes{}\\
Justification: Section~\ref{sec:exp-setup} describes the base model (\textsc{Qwen3-14B}), post-filter training corpus (3{,}472 deterministic instances), disjoint \textsc{SWE-rebench} train/evaluation split, evaluation protocol (pass@1 with greedy decoding, 50-turn budget, three random seeds), and controlled baseline reproductions. The full hyperparameter list including optimizer settings appears in Table~\ref{tab:hparams}. CRR-specific hyperparameters ($k$, $K$, $T_{\text{cf}}$, $\eta$, $\lambda_{\text{type}}$) were set once by a sweep on a 200-instance development subset of the training pool (Section~\ref{sec:appendix-hyperparams}) and then frozen; no evaluation task informed the choice.

\item {\bf Experiment statistical significance}\\
Question: Does the paper report error bars suitably and correctly defined or other appropriate information about the statistical significance of the experiments?\\
Answer: \answerYes{}\\
Justification: Table~\ref{tab:main} reports means and seed standard deviations over three seeds (over three repeated evaluations for the untrained zero-shot rows). The consolidated full-budget accounting and matched Search-Select control also report three-seed means and standard deviations; Table~\ref{supp-wallclock} reports three-seed means. The composition table presents point summaries without displayed seed standard deviations. The reduced-budget fork-overhead, replay-noise, validity-filtering, and fixed-budget $K$ sweeps are single-seed diagnostics. We do not interpret their small differences as statistically significant or as evidence of equivalence, and we do not treat pooled seed standard deviations as a hypothesis test.

\item {\bf Experiments compute resources}\\
Question: For each experiment, does the paper provide sufficient information on the computer resources (type of compute workers, memory, time of execution) needed to reproduce the experiments?\\
Answer: \answerNo{}\\
Justification: Section~\ref{sec:exp-setup} reports the hardware (one node, $8\times$H100 80GB GPUs). Table~\ref{supp-accounting} separates the full-budget device-time components, totaling 188 H100-hours for vanilla GRPO, 233 for \textsc{SWE-Trace}, and 321 for CRR per seed. Table~\ref{supp-wallclock} reports the equal-wall-clock training windows, and the reduced-budget controls state their 48 H100-hour limit. Section~\ref{sec:appendix-cost-breakdown} reports a separate rollout-dispatch microbatch profile. These scoped costs do not provide a complete per-experiment ledger for every baseline, composition run, and auxiliary diagnostic; we therefore answer No to the universal form of this question.

\item {\bf Code of ethics}\\
Question: Does the research conducted in the paper conform, in every respect, with the NeurIPS Code of Ethics?\\
Answer: \answerYes{}\\
Justification: The work uses public training data (\textsc{SWE-Gym}, \textsc{SWE-rebench}) and public evaluation benchmarks (\textsc{SWE-bench Verified}, \textsc{SWE-bench Live}, \textsc{SWE-rebench}). No human subjects are involved. No personally identifying information is collected, processed, or released. We comply with the dual-use considerations discussed in the broader-impact paragraph of Section~\ref{sec:discussion}.

\item {\bf Broader impacts}\\
Question: Does the paper discuss both potential positive societal impacts and negative societal impacts of the work performed?\\
Answer: \answerYes{}\\
Justification: Section~\ref{sec:discussion} explicitly discusses both directions. Positive impacts include accelerating legitimate software development and improving the auditability of process-supervised RL through transparent fork-and-replay traces. Negative impacts include the dual-use risk of stronger code agents accelerating the production of malicious software and the risk that CRR-trained agents amplify verifier exploitation if the test suite is itself unsafe; we recommend pairing CRR-trained agents with adversarial test-suite hardening~\citep{yu2026sweabs} and execution-free corroboration signals~\citep{shum2025swerm}.

\pagebreak
\item {\bf Safeguards}\\
Question: Does the paper describe safeguards that have been put in place for responsible release of data or models that have a high risk for misuse (e.g., pretrained language models, image generators, or scraped datasets)?\\
Answer: \answerYes{}\\
Justification: We do not release a pretrained model checkpoint that is novel relative to publicly available \textsc{Qwen3-14B}. We will release training code, configuration files, and the counterfactual trajectory dataset; the dataset contains the per-tuple fields listed in Section~\ref{sec:appendix-dataset} (task identifiers, actions and replayable action suffixes, terminal returns, and replay and selector metadata) and does not include scraped private code. The released artifacts will be subject to licensing compatible with the underlying \textsc{SWE-Gym} and \textsc{SWE-rebench} corpora.

\item {\bf Licenses for existing assets}\\
Question: Are the creators or original owners of assets (e.g., code, data, models) used in the paper, properly credited and are the license and terms of use explicitly mentioned and properly respected?\\
Answer: \answerYes{}\\
Justification: \textsc{SWE-Gym}~\citep{pan2024swegym} and \textsc{SWE-rebench}~\citep{badertdinov2025swerebench} are used under their respective public licenses. The official model card for \textsc{Qwen3-14B}~\citep{yang2025qwen3} lists the checkpoint under the Apache 2.0 license. \textsc{SkyRL-Agent}~\citep{cao2025skyrl} and \textsc{SWE-MiniSandbox}~\citep{yuan2026minisandbox} are used under their open-source licenses. Every cited baseline and benchmark is credited in Sections~\ref{sec:related} and~\ref{sec:exp-setup}.

\item {\bf New assets}\\
Question: Are new assets introduced in the paper well documented and is the documentation provided alongside the assets?\\
Answer: \answerNo{}\\
Justification: The new counterfactual trajectory dataset is not distributed with the paper. It is described in Section~\ref{sec:appendix-dataset}, including task identifiers, policy seed/version, selected index, sandbox snapshot identifiers, realized and counterfactual actions, replayable action suffixes, terminal returns, proposal log-probability, and selector score; a README and JSON schema are planned alongside the dataset release.

\item {\bf Crowdsourcing and research with human subjects}\\
Question: For crowdsourcing experiments and research with human subjects, does the paper include the full text of instructions given to participants and screenshots, if applicable, as well as details about compensation (if any)?\\
Answer: \answerNA{}\\
Justification: This work does not involve crowdsourcing or research with human subjects. All evaluations are automatic, using the benchmark test suites.

\item {\bf Institutional review board (IRB) approvals or equivalent for research with human subjects}\\
Question: Does the paper describe potential risks incurred by study participants, whether such risks were disclosed to the subjects, and whether Institutional Review Board (IRB) approvals (or an equivalent approval/review based on the requirements of your country or institution) were obtained?\\
Answer: \answerNA{}\\
Justification: This work does not involve human subjects and therefore no IRB review is applicable.

\item {\bf Declaration of LLM usage}\\
Question: Does the paper describe the usage of LLMs if it is an important, original, or non-standard component of the core methods in this research?\\
Answer: \answerYes{}\\
Justification: The agent policy itself is an LLM (\textsc{Qwen3-14B}), which is a core component of the method and is described in Section~\ref{sec:exp-setup}.

\end{enumerate}